\documentclass[pre,twocolumn,amssymb,aps,showpacs,10pt,floatfix]{revtex4-1}
\usepackage[utf8]{inputenc}
\usepackage{xcolor}
\usepackage{amsmath,amssymb,mathrsfs} 
\usepackage{graphicx}
\usepackage{amsfonts}
\usepackage{dsfont}
\usepackage{amssymb}
\usepackage{bm}%
\usepackage{verbatim}
\usepackage{pstricks,pst-node}
\usepackage{amsthm}
\newtheorem{theorem}{Theorem}
\newtheorem{corol}[theorem]{Corollary}
\newtheorem{prop}[theorem]{Proposition}
\newtheorem*{remark}{Remark}

\usepackage[toc,page]{appendix}
\usepackage{multirow}
\usepackage{booktabs}
\newcommand{\ket}[1]{| #1 \rangle}
\newcommand{\bra}[1]{\langle #1 |}
\newcommand{\braket}[2]{\langle #1 | #2 \rangle}
\newcommand{\pket}[1]{[ #1 ]}

\newcommand{\Tr}{\text{Tr}}

\newcommand{\eg}{\hbox{\em e.g.{}}}
\newcommand{\etc}{\hbox{\em etc.{}}}
\newcommand{\ie}{\hbox{\em i.e.{}}}

\newcommand{\Ps}{\mathbb{P}}
\newcommand{\Hs}{\mathcal{H}}
\newcommand{\SSC}{S^2_\text{\rm SC}}

\makeatletter
\g@addto@macro\bfseries{\boldmath}
\makeatother
\begin{document}
\title{Geometry of spin coherent states}
\author{C.{} Chryssomalakos, E.{} Guzm\'an-Gonz\'alez, and E.{} Serrano-Ens\'astiga}
\affiliation{%
Instituto de Ciencias Nucleares, Universidad Nacional Aut\'onoma de M\'exico, Apartado Postal 70-543, Ciudad de M\'exico 04510, M\'exico.%
}
%
%
%
%
%
%
%
%

%
%
%
%

%
%
%
%
%
\begin{abstract}
Spin states of maximal projection along some direction in space are called (spin) coherent, and are, in many aspects, the ``most classical'' available. For any spin $s$, the spin coherent states form a 2-sphere in the projective Hilbert space $\Ps$ of the system. We address several questions regarding that sphere, in particular its possible intersections with complex lines. We also find that, like Dali's iconic clocks, it extends in all possible directions in $\Ps$.  We give a simple expression for the Majorana constellation of the linear combination of two coherent states, and use Mason's theorem to give a lower bound on the number of distinct stars of a linear combination of two arbitrary spin-$s$ states. Finally, we plot the image of the spin coherent sphere, assuming light in $\Ps$ propagates along Fubini-Study geodesics. We argue that, apart from their intrinsic geometric interest, such questions translate into statements experimentalists might find useful.
\end{abstract}
\maketitle

\section{Introduction}

%
%
%
%
%

Quantum theory's predominantly algebraic beginnings have given way, in the last decades, to an intense interest in its geometric aspects. Berry's discovery of geometric phases, and their description as holonomies in a principal bundle, fueled a  renaissance of the theory that continues to our days, further impulsed by advances in quantum computing. Although quantum dynamics has also benefited by this trend (see, \eg, \cite{Ana:91,Sch:96,Ash.Sch:99,Bro.Hug:00}), it is mostly kinematical considerations that drive the field, the principal object of study being the space of quantum states, particularly in its finite dimensional incarnation. 
Properties of states, like entanglement, that are deemed essential for quantum information processing, are seen to admit natural characterizations in purely geometrical terms~\cite{Dur.Vid.Cir:00,Kus.Zyc:01,Mos.Dan:01,Wei.Gol:03,Hub.Kle.Wei.Gon.Guh:09,Aul.Mar.Mur:10,Che.Aul.Haj:14,Mar.Gir.Bra.Bra.Bas:10,Mar.Gir.Bra.Bra.Bas:10,Bag.Bas.Mar:14}, and the gradual assimilation, by the community, of an ever expanding mathematical arsenal (\eg,~\cite{Miy:03,Hey:08,Hol.Luq.Thi:12}) promises to shed a new, bright light on familiar, yet not sufficiently understood concepts. 

Among quantum states, spin coherent (SC) ones~\cite{Rad:71} are the ``most classical'', just like their harmonic oscillator infinite dimensional counterparts, and generalizations thereof. Viewed as totally symmetric $N$-partite systems, they are characterized by their vanishing entanglement, yet, they have been shown to serve in classifying that same quantity as it pertains to other symmetric states~\cite{Man.Cou.Kel.Mil:14}. In paper, a generic spin state can be expanded in a linear combination of appropriate SC states~\cite{Man.Cou.Kel.Mil:14,San.Egu.DiC.Sab.Lam.Sol:16}, while in the laboratory, it can be reconstructed by a knowledge of corresponding transition probabilities~\cite{Ami.Wei:99}, the relation between these two statements being less trivial than one might assume. SC states have also appeared in the characterization of the polarization of light~\cite{Bjo.Gra.Hoz.Leu.San:15,Bjo.Gra.Hoz.Leu.San:15}, and, there too, correspond to maximally classical behavior, that has recently been studied also experimentally~\cite{Bou.etal:17}.

Our own study of SC states, part of which is reported here, revolves around basic questions about the geometry and topology of quantum state space: going beyond the standard folklore, we aim at an intuitive grasp of what ``living in quantum state space'' might be like. For example, it is an elementary fact  that the SC states form a topological 2-sphere, for any value of the spin of the system, but we feel there is much more to know about this, that is simply absent from the literature: assume one stands on a  particular state $\pket{\Psi}$ in quantum state space (we explain our notation in section~\ref{Ron}) 
and looks around, using light that travels along geodesics of the natural Fubini-Study (FS) metric --- what does one see?   Would the SC sphere  look like a distant moon in the sky? Would it look spherical? What part of the sky would it cover? How many times would a light ray intersect its surface, assuming transparency? We do not deny that we would pose these questions in any case for the sheer pleasure of finding out the answer, but it is also true that they have direct physical implications: for example, if looking at the ``SC moon'' from $\pket{\Psi}$  one can see both a front surface and a rear one, this implies that 
$\ket{\Psi}$ (a lift of $\pket{\Psi}$ in the overlying Hilbert space) can be written as a linear combination of two SC states. This, in turn, implies that an experimentalist, equipped with a magnetic field and a beam of particles in a SC state (and a picture of the SC moon taken from $\pket{\Psi}$!), can split the beam in two, rotate one component with the magnetic field to produce a second SC state, and then reconstruct $\ket{\Psi}$ by recombining the two SC states and rotating the superposition in its final orientation. The same comment holds true in the case the seemingly esoteric statement that there is a certain complex line going through $\pket{\Psi}$ and intersecting the SC sphere in two points, is valid.  Formalizing the above discussion, we are led to consider geodesics of the FS metric, and complex lines,  that pass through an arbitrary state $\pket{\Psi}$, and study how they intersect  the SC sphere, as $\pket{\Psi}$ is moved around the quantum state space. We consider most of the questions we pose elementary, but find some of the answers surprising --- the skeptic reader might want to fast forward to figures~\ref{fig:S2SC_1}, \ref{fig:S2SC_2}, on pages~\pageref{fig:S2SC_1}, \pageref{fig:S2SC_2}, and decide whether \emph{that} looks like a spherical moon.

Others before us have explored quantum state space with a similar geometric/visual point a view (see, \eg,~\cite{Kus.Zyc:01,Ben.Bra.Zyc:02,Bro.Gus.Hug:07}) --- the definitive reference in this regard is~\cite{Ben.Zyc:17}, to which, we are glad to admit,  we owe a great deal.
\section{Majorana Constellations}
In a relatively little known 1932 paper~\cite{Maj.Rep}, E.{} Majorana showed how to completely characterize, up to an overall phase, a normalized spin-$s$ state $\ket{\Psi}$ by a set of $2s$ points (\emph{stars}) on the unit sphere, the latter known as the \emph{Majorana constellation} corresponding to $\ket{\Psi}$. The construction generalizes the well known characterization of a spin-1/2 state, up to phase, by a single point on the Bloch sphere. The precise statement is that points in the projective Hilbert space $\Ps =\mathbb{C}P^N$ of a spin-$s$ system ($N \equiv 2s$) are in one-to-one correspondence with unordered sets of (possibly coincident) $2s$ points on the unit sphere. There are various ways to see why this is so --- we mention three that we find most illuminating, starting from Majorana's original construction, and progressing in order of decreasing abstraction.
\subsection{Majorana polynomial of a spin state}
Given an arbitrary spin-$s$ state, expressed in the $S_z$-eigenbasis,
\begin{equation}
|\Psi\rangle=\sum_{m=-s}^s c_m |s,m\rangle
\, ,
\label{sstate}
\end{equation} 
we associate to it its \emph{Majorana polynomial} \cite{Maj.Rep} $p_{\ket{\Psi}}(\zeta)$,
\begin{equation}
p_{|\Psi\rangle}(\zeta)
=
\sum_{m=-s}^{s}
(-1)^{s-m}\sqrt{\binom{2s}{s-m}}
c_m \,  \zeta^{s+m}
\, ,
\label{pjk}
\end{equation}
where $\zeta$ is an auxiliary complex variable. 
The $N$ roots $\zeta_i \in \mathbb{C}$, $i=1,\ldots,N$, of $p_{\ket{\Psi}}$ can be mapped to $N$ points $n_i$  on the $2$-sphere via stereographic projection from the south pole. The resulting constellation, made up of the $N$ stars, is the \emph{stellar representation} of the state $\ket{\Psi}$. If the polynomial turns out of a lower degree, \ie, if $c_{m}=0$ for $m=s, s-1, \ldots, s-k$, then $\zeta=\infty$ is considered a root of multiplicity $k+1$, resulting in the appearance of $k+1$ stars at the south pole of $S^2$. In the rest of this article, a state with stars $\{ n_k \}_{k=1}^{N}$ is denoted by $\ket{n_1, n_2 , \dots, n_N}$ --- note that the ordering of the stars is immaterial. The particular choice of coefficients in~(\ref{pjk}) results in that a transformation $D(\mathsf{R})$ of $|\Psi\rangle$ in Hilbert space, where $D(\mathsf{R})$ is the spin-$s$ irreducible representation of $\mathsf{R} \in SU(2)$, corresponds to a rotation $\mathsf{R}$ of the corresponding constellation on $S^2$.
\subsection{Spin-$s$ state from spin-1/2 constituents}
It is well known that the spin-$s$ state space is mathematically equivalent to the totally symmetric sector of the $2s$-fold tensor power of the spin-1/2 state space. In other words, even though a particular spin-$s$ system might owe its angular momentum to, say, a pair of particles orbiting each other, the properties, in a certain state of the system, under rotations, are  indistinguishable from those of a system of $2s$ spin-1/2 particles, in a particular, totally symmetric (under exchange of any pair of particles) state. The latter can always be obtained by considering first a separable state $\ket{\hat n_1} \otimes \ldots \otimes \ket{\hat n_N}$, where $\ket{\hat{n}}$ is a spin-1/2 state, and subsequently symmetrizing it by summing over all permutations of the particles in the available tensor factors, to obtain the totally symmetric state $\ket{\Psi}$,
\begin{align}
\ket{\Psi}
&=
\ket{n_1,\ldots,n_N}
\nonumber
\\
 &=
 \frac{A_\Psi}{N!} \sum_{\sigma \in S_N} \ket{\hat n_{\sigma_1}} \otimes \ldots \otimes \ket{\hat n_{\sigma_N}}
\label{psisymm}
\, ,
\end{align}
where $A_\Psi$ is a normalization factor,
\begin{equation}
\label{APsidef}
A_\Psi^2=\frac{N!}{\sum_{\sigma \in S_N} \braket{\hat{n}_1}{\hat{n}_{\sigma_1}} 
\ldots 
\braket{\hat{n}_N}{\hat{n}_{\sigma_N}}}
\, ,
\end{equation}
and $S_N$ is the permutation group of $N$ objects. Thus, any spin-$s$ state is equivalent to a state $\ket{\Psi}$ as in~(\ref{psisymm}), and the Majorana constellation of the former is the set of unit vectors $\{n_i\}$, $i=1,\ldots,N$, that appear in~(\ref{psisymm}). 
\subsection{An operational definition}
The above considerations  lead us to an operational definition of the $N$ (possibly coinciding) directions $n_i$ associated to an arbitrary spin-$s$ state $\ket{\Psi}$ (see, \eg,~\cite{Ami.Wei:99}). Whether or not the system in question is made up of spin-1/2 particles, we may use the representation of $\ket{\Psi}$ in~(\ref{psisymm}) to conclude that there are, in general,  $N$ directions in space, such that if a Stern-Gerlach apparatus is pointed along them, the probability of measuring the minimal spin projection $-s$ is zero. Indeed, if the apparatus is pointed to an arbitrary direction $m$, and the total spin projection in that direction is measured to be $-s$, this means that, in the constituent spin-1/2 picture, each spin-1/2 was measured to have projection $-1/2$. If now the direction $m$ coincides with one of the stars $n_i$, the probability that that particular spin, which ``points along'' $n_i$, will project to $-1/2$ is zero (since $\ket{\hat n_i}$ and $\ket{ - \! \hat{n}_i}$ are orthogonal states), and hence the probability that the total projection of the system is measured to be  $-s$ is also zero. In fact, that same reasoning reveals that if $n_i$ has multiplicity $k$ (\ie, there are $k$ stars coinciding there), then a measurement of the spin projection along that direction has zero probability of producing any of the values $-s$, $-s+1$, \ldots, $-s+k-1$.

In view  of the above, it should not come as a surprise that the inner product $\braket{-n}{\Psi}$  is actually proportional to the Majorana polynomial  $p_{\ket{\Psi}}(\zeta)$, where  $\zeta=\tan\frac{\theta}{2}e^{i \phi}$  is the stereographic image of $n=(\sin\theta \cos\phi,\sin\theta \sin\phi,\cos\theta)$.
\section{Spin Coherent States}

%
%
%
%
%
%
%
\subsection{Remarks on notation}
\label{Ron}
Before we delve into our main object of study, we explain our substantially simplified notation.  General directions in physical $\mathbb{R}^3$ will be denoted by $c$, $n$, $m$, \etc, while $x$, $y$, $z$ will be reserved for the cartesian axes. In what follows we deal with \emph{spin-$s$ coherent states}, which are characterized by having a maximal projection $s$ along a particular direction $n$. We denote such states by $\ket{n}$ --- the spin $s$ of the system will be obvious from the context, but still reserve the symbol $\ket{\hat{n}}$ for the spin-1/2 states, so that we can write without problems the formula $\ket{n}=\ket{\hat{n}} \otimes \ldots \otimes \ket{\hat{n}}$.
 Our discussion takes place either in the Hilbert space $\Hs^{N+1}$ of the system, or in the projective Hilbert space $\Ps(\Hs^{N+1})\equiv \Ps^N$, where points are equivalence classes of normalized states differing by a phase factor, and can be identified with the corresponding density matrix. We often omit the superindex when the dimension of the spin state space is clear from the context. Accordingly, we denote the projection of the state $\ket{\Psi} \in \Hs$ by $\pket{\Psi}$ or $\rho_\Psi \in \Ps$. Finally, the Majorana constellation of a state $\pket{\Psi}$ will be given as a list of unit vectors $\{n_1,\ldots,n_N\}$, or of their corresponding (via stereographic projection) complex numbers, say, $\{ \gamma_1, \ldots, \gamma_N \}$.
\subsection{Majorana polynomial as a transition amplitude}
We begin by clarifying the relation, alluded to above, between $p_{\ket{\Psi}}(\zeta)$ and $\braket{-n}{\Psi}$. For a spin-1/2 state we have
\begin{align}
\label{spin12state}
\ket{\hat n}
&=
\cos\frac{\theta}{2}\ket{\hat{z}}+e^{i \phi}\sin\frac{\theta}{2}\ket{\! - \! \hat{z}}
\\
\ket{\! - \! \hat n}
&=
\sin\frac{\theta}{2}\ket{\hat{z}}-e^{i \phi}\cos\frac{\theta}{2}\ket{\! -\! \hat{z}}
\, ,
\end{align}
so that
\begin{align}
\label{ipnni}
\braket{- \hat n}{\hat n_i}
&=
\cos\frac{\theta}{2} \cos\frac{\theta_i}{2}e^{-i\phi} (\zeta-\zeta_i)
\\
&=
\cos\frac{\theta}{2}e^{-i\phi} 
\frac{\zeta-\zeta_i}{\sqrt{1+|\zeta_i|^2}}
\, .
\end{align}
Using~(\ref{psisymm}) for $\ket{\Psi}$ we then find
\begin{equation}
\label{ipPsimn}
\braket{-n}{\Psi}
=
A_\Psi 
\left(
\cos\frac{\theta}{2}e^{-i\phi} 
\right)^{N}
\prod_{i=1}^{N}
\frac{\zeta-\zeta_i}{\sqrt{1+|\zeta_i|^2}}
\, .
\end{equation}
On the other hand, the coefficient $c_s$ of the maximal power of $\zeta$ in $p_{\ket{\Psi}}(\zeta)$ is equal to $\braket{z}{\Psi}$, so that
\begin{equation}
p_{\ket{\Psi}}(\zeta) 
= 
\braket{z}{\Psi}\prod_{i=1}^N (\zeta - \zeta_i) 
= 
A_{\Psi} \prod_{i=1}^N \frac{(\zeta-\zeta_i)}{\sqrt{1+ |\zeta_i|^2}} 
\, ,
\end{equation}
since $\braket{\hat{z}}{\hat{n}_i}=\cos \theta_i/2=(1+|\zeta_i|^2)^{-1/2}$. Comparing the last two equations we arrive at
\begin{equation}
\label{PsinMajrel}
\braket{- n}{\Psi}
=
\left(
\cos \frac{\theta}{2} e^{-i\phi}
\right)^N
p_{\ket{\Psi}}(\zeta)
\, .
\end{equation}
\subsection{SC bases}
The stellar representation of the SC state $\ket{n}$ consists of $N$ coincident stars in the direction $n$. For any $s$, 
the set of SC states is topologically a 2-sphere~\cite{Ben.Zyc:17}, which we denote by $S^2_{\text{SC}}$, 
sitting inside the full projective space $\Ps$. The unit operator may be resolved in SC states,  
$\mathds{1}= (2s+1) \int \ket{n}\bra{n} d\Omega/4\pi$, implying that any state  can be written as a 
(continuously infinite) linear combination of SC states. The following theorem shows that, in fact, 
\emph{any} $N+1$ SC states  will do --- the theorem may be found in the 
supplementary material to~\cite{San.Egu.DiC.Sab.Lam.Sol:16}, we give nevertheless a (slightly more streamlined) proof below, to establish our notation, and so that we can refer to intermediate results in the rest of the paper.
\begin{theorem}
\label{SCbasis_thm}
 Any set of $N+1$ distinct SC states $\{ \ket{c_k} \}_{k=0}^N$ forms a basis in the Hilbert space $\Hs^{N+1}$.
\end{theorem}
\begin{proof}
 Let $\gamma_k$ be the  complex number associated, via stereographic projection,  to the direction $c_k$ and let $\ket{\Psi}$ be an arbitrary state with $N$ associated complex numbers $\{ \zeta_k \}_{k=1}^N$ (\ie, the $\zeta_k$'s are the roots corresponding to the stars of $\ket{\Psi}$). The expansion
\begin{equation}
\label{SC.bas}
\ket{\Psi}=\sum_{k=0}^N \alpha'^{k} \ket{c_k} 
\end{equation}
implies the following relation for the corresponding Majorana polynomials 
\begin{equation}
\prod_{j=1}^N \frac{(\zeta- \zeta_j )}{\left( 1+ |\zeta_j|^2  \right)^{1/2}} 
= 
\sum_{k=0}^N \frac{\alpha^k}{\left(1+ |\gamma_k|^2 \right)^{N/2}} (\zeta- \gamma_k)^N \, ,
\label{SCcoefficients}
\end{equation}
with $\alpha^k =  \alpha'^k A_{n_k} / A_{\Psi} =\alpha'^k / A_{\Psi}$, and where  the $A$'s are the normalization factors introduced in~(\ref{psisymm}), (\ref{APsidef})  (note that for an SC state $\ket{n}$, $A_n=1$). 
Expanding each side we obtain
\begin{equation}
\begin{split}
\sum_{j=0}^N (-1)^{N-j}& \zeta^j b_{N-j}\\
 &= 
 \sum_{j=0}^N 
 (-1)^{N-j} \zeta^j \binom{N}{j} 
 \sum_{k=0}^N  
 \tilde{\alpha}^k \gamma_k^{N-j}
\end{split} 
\, ,
\end{equation}
where
\begin{equation}
\label{factor.alpha}
\tilde{\alpha}^k 
= 
\frac{\alpha^k}{\left(1+ |\gamma_k|^2 \right)^{N/2}} 
\left( 
\prod_{m=1}^N \left( 1+ |\zeta_m|^2 \right)^{1/2} 
\right) 
\, ,
\end{equation}
and $b_j$ are the symmetric polynomials of the numbers $\{ \zeta_k \}_{k=1}^N$, with $b_N=\prod_{i=1}^N \zeta_i$ and $b_0=1$. Comparing the powers of 
$\zeta$ on both sides in (\ref{factor.alpha}), we obtain the following system of equations
\begin{equation}
\label{s.van}
\left(
\begin{array}{c c c c}
1 & 1 & \dots & 1
\\
\gamma_0 & \gamma_1 & \ddots & \gamma_N
\\
\vdots & \ddots & \ddots & \vdots
\\
\gamma_0^N & \gamma_1^N & \dots & \gamma_N^N 
\end{array}
\right)
\left( 
\begin{array}{c}
\tilde{\alpha}^0 
\\
\tilde{\alpha}^1 
\\
\vdots
\\
\tilde{\alpha}^N
\end{array}
\right)
=
\left( 
\begin{array}{c}
 \tilde{b}_0
\\
 \tilde{b}_1
\\
\vdots
\\
 \tilde{b}_N
\end{array}
\right)
 \, ,
\end{equation}
with $\tilde{b}_j = \binom{N}{j}^{-1} b_j$. The matrix in the left hand side of the above equation, which we will denote by  $\mathsf{V}$, is of  the Vandermonde form, and is invertible if and only if all the numbers $\gamma_i $ are distinct.
\end{proof}
\noindent Using the known formula for the inverse of a Vandermonde matrix~\cite{Vander.inv} we find that $(\mathsf{V}^{-1})_{ij}$ is the coefficient of the term $\zeta^j$ in the polynomial $P_{i}(\zeta) / P_i (\gamma_i)$, where $P_i(\zeta)=\prod_{k=0, k \neq i}^{N}(\zeta-\gamma_k)$ (the indices in these formulas run from 0 to $N$). We discuss now some ramifications of the theorem, before giving a series of examples.
%
%
%
%
%
%
%
\subsection{Dual basis}
\indent
Note that, according to theorem~\ref{SCbasis_thm},  \emph{any} set of $N+1$ SC states forms a basis, without any restriction whatsoever on their relative positions, proximity, \emph{etc}. Given such a basis $\{\ket{c_k} \}_{k=0}^N$, we denote by $\{\ket{c^k} \}_{k=0}^N$  its dual basis, such that 
\begin{align}
\braket{c^j}{c_i}
&=
 \delta^j_{\phantom{j}i}
 \label{duality}
 \\
\sum_{i=0}^N \ket{c^i} \bra{ c_i}
&= 
\mathds{1}
\label{resid}
\, .
\end{align}
Note that, in general, $\braket{c^i}{c^i} \neq 1$. 
It is easy to see that  a spin state 
\begin{equation}
\label{Psin}
\ket{\Psi}=\ket{n_1,\ldots,n_N}
\, ,
\end{equation}
is orthogonal to any SC state whose  direction is antipodal  to one of the stars associated to $\ket{\Psi}$,
\begin{equation}
\label{PsiorthoSC}
\braket{-n_i}{\Psi}=0
\, .
\end{equation}
In particular, if $\ket{\Psi}$ has no degeneracy (\ie, coincident stars), it is orthogonal to $N$ SC states, which shows that the dual basis element  $\ket{c^i}$ is given by
\begin{equation}
\label{cidef}
\ket{c^i} 
=  
\frac{\ket{- \! c_0 , -c_1 , \dots , \widehat{-c_{i}},  \dots , -c_{N}} }{
\braket{\phantom{\rule{0ex}{2.3ex}} c_i}{- \! c_0 , -c_1 , \dots , \widehat{-c_{i}},  \dots , -c_{N}}
} 
\, ,
\end{equation}
where the wide hat denotes omission. Note that the denominator in~(\ref{cidef}) is nonzero, since $\ket{\hat{c}_i}$ is only orthogonal to $\ket{\! - \! \hat{c}_i}$ and no other spin-1/2 state. We remark also that~(\ref{resid}) implies that 
\begin{equation}
\label{expansionadapt}
\ket{\Psi} =  \braket{c^i}{\Psi} \ket{c_i}
\, ,
\end{equation}
which is an alternative to solving~(\ref{s.van}).
%
%
%
%
%
%
%
\subsection{Extrema of the Husimi function and adapted SC bases}
We wish now to associate to a generic state $\ket{\Psi}$ (\ie, one with $N$ distinct stars), an \emph{adapted SC basis} $\{ \ket{c_i} \}_{i=0}^N$, the elements of which, as the name suggests, are all SC states. For $\ket{\Psi}$ as in~(\ref{Psin}), the elements $\{ c_i \}_{i=1}^N$ are just given by the stars of $\ket{\Psi}$, $c_i=n_i$, $i=1,\ldots,N$. The remaining element $\ket{c_0}$ is defined as follows: there is a unique 1D linear subspace that is orthogonal to all $\ket{n_i}$, $i=1,\ldots, N$. In fact, it consists of the complex multiples of the state $\ket{\tilde{\Psi}}$ \emph{antipodal} to $\ket{\Psi}$, \ie, the state  whose stars are antipodal to those of $\ket{\Psi}$, $\ket{\tilde{\Psi}}=\ket{\! - \! n_1, \dots, -n_N}$. This state is not itself SC, but has, generically, a single closest SC state, in the FS metric of $\Ps$ --- this latter state is chosen as $\ket{c_0}$. Central inversion of a constellation is an isometry for the FS metric, so that if $\ket{c_0}$ is the closest SC state to $\ket{\tilde{\Psi}}$, the SC state $\ket{\! - \! c_0}$ is the one closest to  $\ket{\Psi}$. 
The above may be summarized neatly as follows:  \emph{the $N+1$ elements of the SC basis adapted to a generic state $\ket{\Psi}$  are defined by the antipodes of the extremal points (one maximum and $N$ minima) of its Husimi function~\cite{Ben.Zyc:17} $H_{\Psi}(n)=|\braket{n}{\Psi}|^2$}. 

We derive now a necessary and sufficient condition for an SC state $\ket{n_0}$ to be closest to a  generic state $\ket{\Psi}$. SC states $\ket{n}$, nearby $\ket{n_0}$,  can be obtained by a rotation, 
\begin{equation}
\label{nfromn0}
\ket{n}=R\ket{n_0}=e^{-i b S_-}e^{-i a S_z}e^{-i c S_+} \ket{n_0}
\, ,
\end{equation}
where the reference frame has been rotated so as to make $n_0$ coincide with $z$, and $a$, $b$, $c$ $\in \mathbb{C}$ are functions of the rotation parameters (see, \emph{e.g.}, the supplementary material in~\cite{Gir.Bra.Bag.Bas.Mar:15}). With $S_+ \ket{n_0}=0$ and $S_z \ket{n_0}=s \ket{n_0}$, we get
\begin{equation}
\label{Husimiex1}
H_{\Psi}(n) 
= 
e^{2s \Im (a)} \bra{\Psi} e^{-i b S_-} \ket{n_0} \bra{n_0} e^{i \bar{b} S_+} \ket{\Psi} 
\, ,
\end{equation}
where $\Im$ denotes imaginary part.
Taking the derivative with respect to $b$, and setting it equal to zero, at $b=0$, gives 
\begin{equation}
\label{extremecond}
\bra{\Psi}S_-\ket{n_0} \braket{n_0}{\Psi} = 0
\, .
\end{equation}
When the second factor in the left hand side above vanishes, $\ket{n_0}$ is orthogonal to $\ket{\Psi}$, and we get an SC state at maximal distance (equal to $\pi/2$) from $\ket{\Psi}$ --- this only happens for $n_0$ antipodal to any of the stars of $\ket{\Psi}$. For $\ket{n_0}$ to be closest to $\ket{\Psi}$ the first factor must vanish, implying that
\begin{equation}
\label{Hus.f}
\braket{n_0 , s-1}{\Psi} = 0 
\, , 
\quad  
\ket{n_0, s-1} \equiv  \ket{-n_0, n_0 , \dots , n_0 } 
\, ,
\end{equation}
where $(n_0 \cdot S) \ket{n_0,k}=k \ket{n_0,k}$.
We turn now to a  characterization of the nature of the critical points of the Husimi function $H_\Psi$.
\begin{theorem}
Consider a critical point $n_0$ of the Husimi function  $H_\Psi$ and expand $\ket{\Psi}$ in the $n_0 \cdot S$ eigenbasis, $\ket{\Psi}=\sum_{m=-s}^s \rho_m e^{i \alpha_m} \ket{n_0,m}$. Then
\begin{enumerate}
\item  If $\rho_s=0$ then $n_0$ is a global minimum of $H_\Psi$ ( $\pket{n_0}$ is at a maximal distance from $\pket{\Psi}$).
\item If  
$\rho_{s-1}=0$, and 
$\sqrt{s} \rho_s>\sqrt{2s-1} \rho_{s-2}$,
then $n_0$ is a local maximum of $H_\Psi$ ( $\pket{n_0}$ is at a minimal distance from $\pket{\Psi}$).
\item If 
$\rho_{s-1}=0$, and 
$\sqrt{s} \rho_s< \sqrt{2s-1} \rho_{s-2}$,
then $n_0$ is a saddle point of $H_\Psi$: moving along the $\phi=(\alpha_{s-2}-\alpha_s)/2 \mod \pi$ direction on the sphere $n_0$ is a local minimum, while in the orthogonal direction it is a local maximum.
\end{enumerate}
\end{theorem}
\begin{proof}
If $\rho_s=|\braket{n_0}{\Psi}|=0$ then $H_\Psi$ attains a global minimum at $n_0$  since it is  either positive or zero. 
For the other critical points, the condition for criticality is $\rho_{s-1}=|\braket{n_0,s-1}{\Psi}|=0$, as we have already proved.
To further characterize the critical points, we will expand the Husimi function around them, up to second order in the angular distance. 
Assume, as before,  that $n_0$ is along the $z$ axis
and consider an SC state close to $\ket{z}$ characterized by the angles $\theta$ and $\phi$. Then we have, up to second order in $\theta$,
\begin{align*}
H_\Psi
&=
\left|\bra{\Psi}e^{-i S_z \phi} e^{-i \theta S_y }\ket{z}\right|^2
\\
&=
\left|
\bra{\Psi}
\left (
( 1-\frac{s\theta^2}{4}) \ket{z}
\right. \right.
\\
&
\left. \left.
\quad +\frac{\theta^2}{4} e^{2i\phi} \sqrt{s(2s-1)} \ket{z,s-2}
\right)
\right|^2
\\
&=
(1-\frac{s\theta^2}{2}) \rho_s^2
\\
&
\quad
+\frac{\theta^2}{2} \sqrt{s(2s-1)}\rho_{s-2} \rho_s \cos(2\phi \! - \!  \alpha_{s-2} \! + \! \alpha_s) 
\\
&=
H(z)-\frac{\theta^2}{2} \Delta
\, ,
\end{align*}
where
\begin{equation*}
\Delta 
\equiv
s \rho_s^2
- \sqrt{s(2s-1)} \rho_{s-2} \rho_s \cos(2\phi - \alpha_{s-2} + \alpha_s)
\, ,
\end{equation*}
and  $\rho_{s-1}=0$ was used.  In order for $H_\Psi$ to have a local maximum, it is necessary for 
$\Delta$ to be positive for all $\phi$.
On the other hand, the minimum value of $\Delta$ is obtained when $2\phi-\alpha_{s-2}+\alpha_s=2k\pi$, $k \in \mathbb{Z}$, 
and for that minimum to be positive it must hold 
\begin{align*}
\sqrt{s}\rho_s >\sqrt{2s-1}\rho_{s-2}
\, ,
\end{align*}
which proves the second case of the theorem. If the previous inequality is reversed the minimum value of $\Delta$ will be negative. Given that its maximal value is evidently positive,  we have a saddle point, and the stated principal directions follow easily. This concludes the last case of the proof.
\end{proof}
It has been shown in~\cite{Man.Cou.Kel.Mil:14} that a spin-$s$ state $\ket{\Psi}$, with maximal star degeneracy less than $\lfloor(N+1)/2\rfloor$, can be written as a linear combination of at most $\lfloor(N+1)/2\rfloor$ SC states, which depend on $\ket{\Psi}$ ($\lfloor \cdot \rfloor$ denotes integer part) --- it would be interesting to explore the relation between that decomposition and ours here, equation~(\ref{expansionadapt}). Note that the moduli of the expansion coefficients $\braket{c^i}{\Psi}$ in that equation are invariant under rotations of $\ket{\Psi}$ --- whether they provide coordinates in the quotient (shape) space $\Ps /SU(2)$ is an open question that we plan on addressing elsewhere.
%
%
%
%
%
%

\subsection{Examples of adapted bases}

%
%
%
%
%
Denote by  $n_i$, $i=1,\ldots,N$, the stars  of the state $\ket{\Psi}$ and by $\zeta_i$ their projections in the complex plane. Similarly, denote by $c_k$,  $k=0,\ldots, N$, the stars of an arbitrary SC  basis, and by $\gamma_k$ their complex projections. Finally, denote by $\alpha_k$, $k=0, \ldots, N$ the expansion coefficients  in (\ref{SCcoefficients}). 
\subsubsection{Adapted SC basis for spin 1/2}
Any two distinct SC  states form a basis in the Hilbert space $\Hs^2$. For $\ket{\Psi}=\ket{\hat n}$ the adapted SC basis is 
$(\ket{\hat c_0}, \ket{\hat c_1})=(\ket{\! \! - \! \! \hat n}$, $\ket{\hat n})$ and the corresponding expansion coefficients are trivially $(\alpha_0,\alpha_1)=(0,1)$.
\subsubsection{Adapted SC basis for spin 1}
For $s=1$, any set of 3 SC states forms a basis in $\Hs^3$. Relations~(\ref{s.van}) become
\begin{equation}
\label{exps1}
\begin{split}
\big(1+ \zeta_1  \bar{\zeta}_1 & \big)^{1/2}
\big(1+ \zeta_2 \bar{\zeta}_2 \big)^{1/2}
\left( 
\begin{array}{c}
\frac{\alpha^0}{1+ \gamma_0 \bar{\gamma}_0}
\\
\frac{\phantom{\rule{0ex}{2.5ex}} \alpha^1}{1+ \gamma_1 \bar{\gamma}_1}
\\
\frac{\phantom{\rule{0ex}{2.5ex}} \alpha^2}{1+ \gamma_2 \bar{\gamma}_2}
\end{array}
\right)
\\
&
\equiv
\left( 
\begin{array}{c}
\tilde{\alpha}^0
\\
\phantom{\rule{0ex}{2.8ex}}
\tilde{\alpha}^1
\\
\phantom{\rule{0ex}{2.8ex}}
\tilde{\alpha}^2
\end{array}
\right)
=
\left(
\begin{array}{c}
\frac{s_{11}s_{22}+s_{12}s_{21}}{2\gamma_{01}\gamma_{02}}
\\
\phantom{\rule{0ex}{2.8ex}}
\frac{ s_{10}s_{22}+s_{12}s_{20}}{2\gamma_{10}\gamma_{12}}
\\
\phantom{\rule{0ex}{2.8ex}}
\frac{ s_{10} s_{21}+s_{11}s_{20}}{2\gamma_{20}\gamma_{21}}
\end{array}
\right) 
\end{split}
\, ,
\end{equation}
where $s_{ij} \equiv \zeta_i-\gamma_j$, and $\gamma_{ij}\equiv \gamma_i-\gamma_j$. We orient the constellation of $\ket{\Psi}$ so that the two stars are in the $x$-$y$ plane, bisected by the $x$ axis. Then, $\zeta_1 = \gamma_1= e^{i \phi}$ and $\zeta_2 = \gamma_2 = e^{-i \phi}$, with $0 < \phi < \pi /2$. The SC state closest  to $\ket{\Psi}$ has its star at $x$, so $\gamma_0=-1$. Relations (\ref{exps1}) give
\begin{equation}
\label{N2.fact}
\left( 
\begin{array}{c}
\alpha^0
\\
\alpha^1
\\
\alpha^2
\end{array}
\right)
=
\left(
\begin{array}{c}
1- \cos \phi
\\
e^{-i \phi}/2
\\
e^{i \phi}/2
\end{array}
\right) 
\, .
\end{equation}
%
%
%
%
%
%
%
%
%
%
%
%
%
\subsubsection{Adapted SC basis for spin 3/2}
Our last example is a state with $s=3/2$. For the general case, the (tilded) expansion coefficients are
\begin{equation}
\left( 
\begin{array}{c}
\tilde{\alpha}^0
\\
\phantom{\rule{0ex}{2.8ex}}
\tilde{\alpha}^1 
\\
\phantom{\rule{0ex}{2.8ex}}
\tilde{\alpha}^2
\\
\phantom{\rule{0ex}{2.8ex}}
\tilde{\alpha}^3
\end{array}
\right)
=
\left(
\begin{array}{c}
\frac{
s_{11} s_{22} s_{33}
+
s_{12} s_{23} s_{31}
+
s_{13} s_{21} s_{32}
}{
3 \gamma_{01} \gamma_{02} \gamma_{03}
}
\\
\phantom{\rule{0ex}{2.8ex}}
\frac{ 
s_{10} s_{22} s_{33} 
+
s_{12} s_{23} s_{30}
+
s_{13} s_{20} s_{32}
}{
3\gamma_{10} \gamma_{12} \gamma_{13}
}
\\
\phantom{\rule{0ex}{2.8ex}}
\frac{
s_{10} s_{21} s_{33}
+
s_{11} s_{23} s_{30}
+
s_{13} s_{20} s_{31}
}{
3\gamma_{20} \gamma_{21} \gamma_{23}
}
\\
\phantom{\rule{0ex}{2.8ex}}
\frac{
s_{10} s_{21} s_{32}
+
s_{11} s_{22} s_{30}
+
s_{12} s_{20} s_{31}
}{
3 \gamma_{30} \gamma_{31} \gamma_{32}
}
\end{array}
\right) \, ,
\end{equation}
where 
\begin{equation}
\left( 
\begin{array}{c}
\tilde{\alpha}^0
\\
\phantom{\rule{0ex}{2.8ex}}
\tilde{\alpha}^1
\\
\phantom{\rule{0ex}{2.8ex}}
\tilde{\alpha}^2
\\
\phantom{\rule{0ex}{2.8ex}}
\tilde{\alpha}^3
\end{array}
\right)
\equiv
\left(
\prod_{i=1}^3
\left(1+ \zeta_i \bar{\zeta}_i \right)^{1/2}
\right)
\left( 
\begin{array}{c}
\frac{\alpha^0}{(1+ \gamma_0 \bar{\gamma}_0)^{3/2}}
\\
\phantom{\rule{0ex}{2.8ex}}
\frac{\alpha^1}{(1+ \gamma_1 \bar{\gamma}_1)^{3/2}}
\\
\phantom{\rule{0ex}{2.8ex}}
\frac{\alpha^2}{(1+ \gamma_2 \bar{\gamma}_2)^{3/2}}
\\
\phantom{\rule{0ex}{2.8ex}}
\frac{\alpha^3}{(1+ \gamma_3 \bar{\gamma}_3)^{3/2}}
\end{array}
\right)
 \, ,
\end{equation}
and with the associated SC basis, $\{\gamma_0, \gamma_i=\zeta_i \}_{i=1}^N$, they reduce to
\begin{equation}
\left( 
\begin{array}{c}
\tilde{\alpha}^0
\\
\phantom{\rule{0ex}{2.8ex}}
\tilde{\alpha}^1 
\\
\phantom{\rule{0ex}{2.8ex}}
\tilde{\alpha}^2
\\
\phantom{\rule{0ex}{2.8ex}}
\tilde{\alpha}^3
\end{array}
\right)
=
\left(
\begin{array}{c}
0
\\
\phantom{\rule{0ex}{2.8ex}}
\frac{-(\gamma_2-\gamma_3)^2}{3(\gamma_1 - \gamma_2)(\gamma_1 - \gamma_3)}
\\
\phantom{\rule{0ex}{2.8ex}}
\frac{-(\gamma_3-\gamma_1)^2}{3(\gamma_2 - \gamma_1)(\gamma_2 - \gamma_3)}
\\
\phantom{\rule{0ex}{2.8ex}}
\frac{-(\gamma_1-\gamma_2)^2}{3(\gamma_3 - \gamma_1)(\gamma_3 - \gamma_2)}
\end{array}
\right)
\, ,
\end{equation}
independent of the choice of $\gamma_0$.
Note that $\alpha^0 = 0$, and $\alpha^1, \alpha^2, \alpha^3$ do not depend of $\gamma_0$. This result generalizes to any half-integer spin state, as the following proposition asserts, and originates in the fact that, for such states, $\bra{\Psi} T\ket{\Psi}=0$, where $T$ is the time-reversal operator, that acts like the antipode map on constellations.
\begin{prop}
The expectation value of the time-reversal operator $T$ in a half-integer spin state vanishes. 
\end{prop}
\begin{proof}
For $s=1/2$, $T=-i \sigma_y K$, where $K$ is the complex conjugate operator --- for higher spins, $T$ is just the tensorial power of this expression. It is easily seen that $T^2 \ket{n_1 , \dots , n_N} = (-1)^N \ket{n_1 , \dots , n_N}$, and $T$ is antiunitary, $\left( T \ket{\Psi_1} , T \ket{\Psi_2}  \right) = \left( \ket{\Psi_2} , \ket{\Psi_1}  \right) $, where we denote the inner product between two states as $(\cdot , \cdot) $. With these properties of $T$ in mind, we compute $(-1)^N (\ket{\Psi}, T \ket{\Psi}) = (T^2 \ket{\Psi}, T \ket{\Psi}) = ( \ket{\Psi}, T \ket{\Psi})$, and therefore, for $N=2s$ odd, $ (\ket{\Psi}, T \ket{\Psi}) =0$. 
\end{proof}
In our case,  for $\ket{\Psi}=\ket{n_1,\ldots,n_N}$, we have $\ket{c^0}\propto \ket{- \! n_1,\ldots,-n_N}$, so that $\alpha^0\propto \braket{c^0}{\Psi}=0$ for $s$ half-integer ($N$ odd).
%
%
%
%
%
%
%
%
%
%
%

\section{Geometrical Aspects of the Spin Coherent Sphere}

%
%
%
%
%
%
\subsection{A Dali 2-sphere}
In this section we study how the 2-sphere of SC states $\SSC$ is immersed in the projective Hilbert space of spin-$s$ states $\Ps$. The restriction of 
the Fubini-Study metric to $S^2_{\text{SC}}$ renders the latter isometric to a euclidean ``round'' 2-sphere. The question we pose is how many independent directions does $\SSC$ explore in $\Ps$? We find that, just like Dali's iconic clocks (topological discs) cannot be contained in any 2-plane, the spin coherent 2-sphere extends in all available directions in $\Ps$. The precise statement is the following
\begin{theorem}
\label{thm:maxdim}
Consider an arbitrary state $\pket{\Psi}$ in $\Ps$ and let $\exp_\Psi$ be the exponential map from $T_{\pket{\Psi}} \Ps$, the tangent space at $\pket{\Psi}$,
to $\Ps$. The inverse image of $\SSC$ under this map, $\log_\Psi(\SSC)$, is of maximal dimension in $T_{\pket{\Psi}} \Ps$.
\end{theorem}
\begin{proof}
We represent $\pket{\Psi} \in \Ps$ by the density matrix $\rho_\Psi=\ket{\Psi}\bra{\Psi}$. Then, a tangent vector in 
$T_{\pket{\Psi}} \Ps$ is represented by the matrix $\ket{\Psi} \bra{\varphi}+\ket{\varphi}\bra{\Psi}$ for some $\ket{\varphi} \in \Hs$ satisfying $\braket{\varphi}{\Psi}=0$.  Suppose that $\log(\SSC)$ is contained in an affine subspace of $T_{\pket{\Psi}} \Ps$ of real dimension lower than 
$4s$.  Then there exists a tangent vector $X$ in $T_{\pket{\Psi}} \Ps$,
\begin{equation}
X=\ket{\Psi}\bra{\chi}+ \ket{\chi}\bra{\Psi}\,,\quad \text{with $\braket{\chi}{\Psi}$=0}
\, ,
\label{eq:Xpsi}
\end{equation}
 such that the inner product between $X$ and $v_n\equiv \log_\Psi \rho_n$, $\rho_n \equiv \ket{n}\bra{n}$,
is constant, say, equal to $\lambda$, for all SC states $\pket{n}$.
Using the explicit expression for $v_n$ in  eqs. (\ref{rhodt0}) and (\ref{logrhon}) below, we  find
\begin{align}
\lambda&=\frac{1}{2}\Tr( v_n X) \nonumber \\
&= \frac{\omega_{n\Psi}}{2 \sin \omega_{n\Psi}} \left( e^{i \eta_{n\Psi}} \braket{\chi}{n}+e^{-i \eta_{n\Psi}} \braket{n}{\chi}\right)
\,,
\label{eq:XprodCoherent}
\end{align}
where  $\braket{n}{\Psi} \equiv \cos \omega_{n\Psi} e^{i \eta_{n\Psi}}$, $\omega_{n\Psi} \in [0,\pi/2]$.
The condition $\eta_{n\Psi}=0$ fixes the phase of all $\ket{n}$, except for the isolated points where $\braket{n}{\Psi}=0$ --- since this latter set is of measure zero, it does not affect our argument below. From~(\ref{eq:XprodCoherent}) we find
\begin{equation}
\Re\braket{\chi}{n}= 
\frac{\sin \omega_{n\Psi}}{\omega_{n\Psi}} \lambda
\, ,
\label{eq:nchireal}
\end{equation}
where $\Re$ denotes real part.
Now we will consider two cases, $\lambda=0$ and $\lambda \neq 0$.
If $\lambda=0$ then, by equation (\ref{eq:nchireal}), $\braket{n}{\chi}$ is imaginary,
\begin{equation}
\braket{n}{\chi}= \cos \omega_{n \chi} e^{i \eta_{n \chi}}=\pm i \cos \omega_{n \chi}
\,.
\end{equation}
The crucial observation at this point is that the real function $f(n)=-i \braket{n}{\chi}=\pm  \cos \omega_{\eta \chi}$, where one of the two possible signs is chosen, cannot change sign on the sphere, since it only has a finite number of isolated zeros. Indeed, assuming that $f$ takes both positive and negative values, one may always choose a curve on the sphere that connects the corresponding points, without passing through any of the isolated zeros of $f$, leading to \emph{absurdum}, as, by the intermediate value theorem, $f$ must have a zero in a certain point of the curve. Having established this fact about $f$, we use the completeness relation for $\ket{n}$ to arrive at
\begin{align}
0
&=
\braket{\Psi}{\chi}
=
\frac{2s+1}{4\pi}\int \braket{\Psi}{n}\braket{n}{\chi} d\Omega
\nonumber\\
&=
\pm i \frac{2s+1}{4\pi}\int \cos \omega_{n\Psi}\cos  \omega_{n\chi} d\Omega
\,,
\label{eq:zeroBraketPsiChi}
\end{align}
which is not possible, since the integrand is non-negative.

If $\lambda\neq 0$, equation (\ref{eq:nchireal}), and the fact  that $\braket{n}{\Psi}$ is real, imply
\begin{equation*}
\Re(\braket{\chi}{n}) \braket{n}{\Psi} 
= 
\frac{\cos \omega_{n\Psi} \sin \omega_{n\Psi}}{\omega_{n\Psi}} \lambda 
\, ,
\end{equation*}
so that
\begin{equation*}
\int \Re(\braket{\chi}{n}) \braket{n}{\Psi} d\Omega 
= 
\lambda \int\frac{\cos \omega_{n\Psi} \sin \omega_{n\Psi}}{\omega_{n\Psi}} d\Omega
\,.
\end{equation*}
By the same argument that lead to (\ref{eq:zeroBraketPsiChi}), the left hand side is zero, but
the  integrand in the right hand side is positive almost everywhere, leading again to  \emph{absurdum}, which shows that 
 such $X$ does not exist, and the proof is complete.
  \end{proof}
\subsection{Closest SC states}
\label{CSCs}
Given a state $\pket{\Psi}$, we would like to know which SC states are closest to it, and which ones are furthest away. We may think of this question in the following terms: consider a geodesic $(4s-1)$-sphere of radius $r$, $S_r$, centered at $\pket{\Psi}$, \ie, the locus of points in $\Ps$ that are a fixed geodesic distance $r$ from $\pket{\Psi}$ . The intersection points of $S_r$ with $\SSC$ give those SC states that are at a distance $r$ from $\pket{\Psi}$. For $\pket{\Psi}$ non-SC, and $r$ sufficiently small, the intersection is null. As $r$ increases, it reaches a critical value $r_\text{c}$ at which $S_{r_\text{c}}$ just touches $\SSC$ at, generically, a single point $\pket{n_0}$. The value of $r_\text{c}$ is the geometrical measure of entanglement of $\pket{\Psi}$ \cite{Bro.Hug:00}. For $r > r_\text{c}$, the intersection is one-dimensional, consisting, generically,  of the union of topological circles. When $r$ reaches its maximal value $\pi/2$, $\SSC$ is tangent to $S_{\pi/2}$ ``from the inside'', touching it at exactly $N$ points, which are the SC states antipodal to the stars (assumed distinct) of $\pket{\Psi}$ --- the collection of these states, lifted arbitrarily in $\Hs$, forms a basis of the orthogonal complement of $\ket{\Psi}$ in $\Hs$.  
\begin{remark}
 The set of states such that the closest SC state is not unique is of measure zero. In fact, this set is at most of dimension $4s-1$. 
\end{remark}
\begin{proof}
 Consider two distinct SC states $\pket{c_1}$ and $\pket{c_2}$ and let $\pket{\Psi}$ be any state such that  $\pket{c_1}$ and $\pket{c_2}$ are both the 
 closest SC states of $\pket{\Psi}$. As was shown in (\ref{Hus.f}) this implies that $\braket{c_1,s-1}{\Psi}=\braket{c_2,s-1}{\Psi}=0$. These are two complex
 equations so that the locus of states that satisfy them has real dimension $4s-4$. Since they must also satisfy the condition 
 $|\braket{c_1}{\Psi}|=|\braket{c_2}{\Psi}|$ for
them to be equidistant, the dimension of all the states whose closest SC state are $\pket{c_1}$ and $\pket{c_2}$ is at most $4s-5$. Finally note that the space 
of the pair of SC states
$\pket{c_1}$ and $\pket{c_2}$ is of dimension $4$. Because of these observations, the space of states where the closest SC state is not unique is of dimension at most $4s-1$ as
claimed. 
\end{proof}
\begin{figure}
\centering
\includegraphics[scale=0.3]{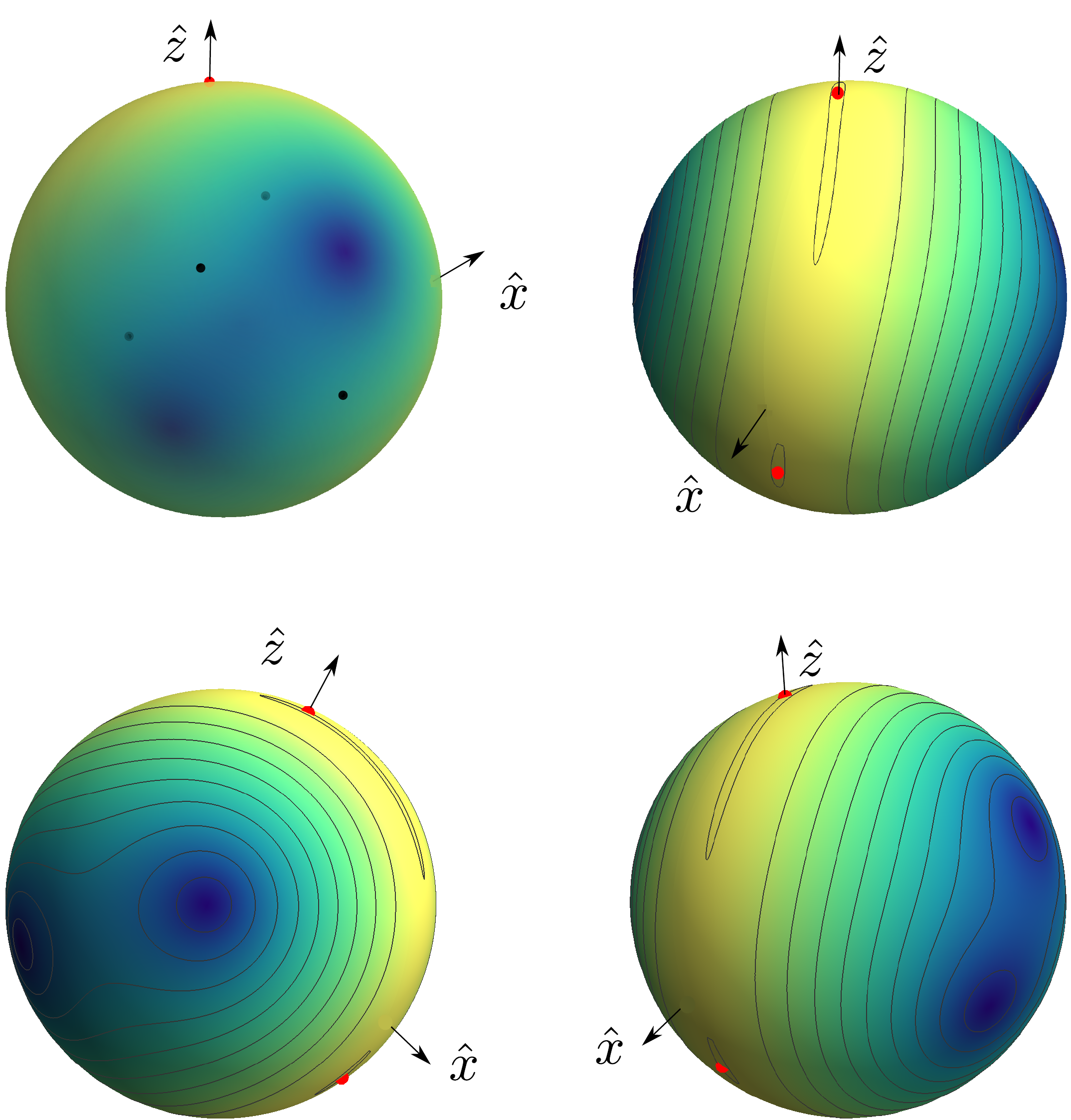}
\\
\vspace{3ex}
\includegraphics[scale=0.2]{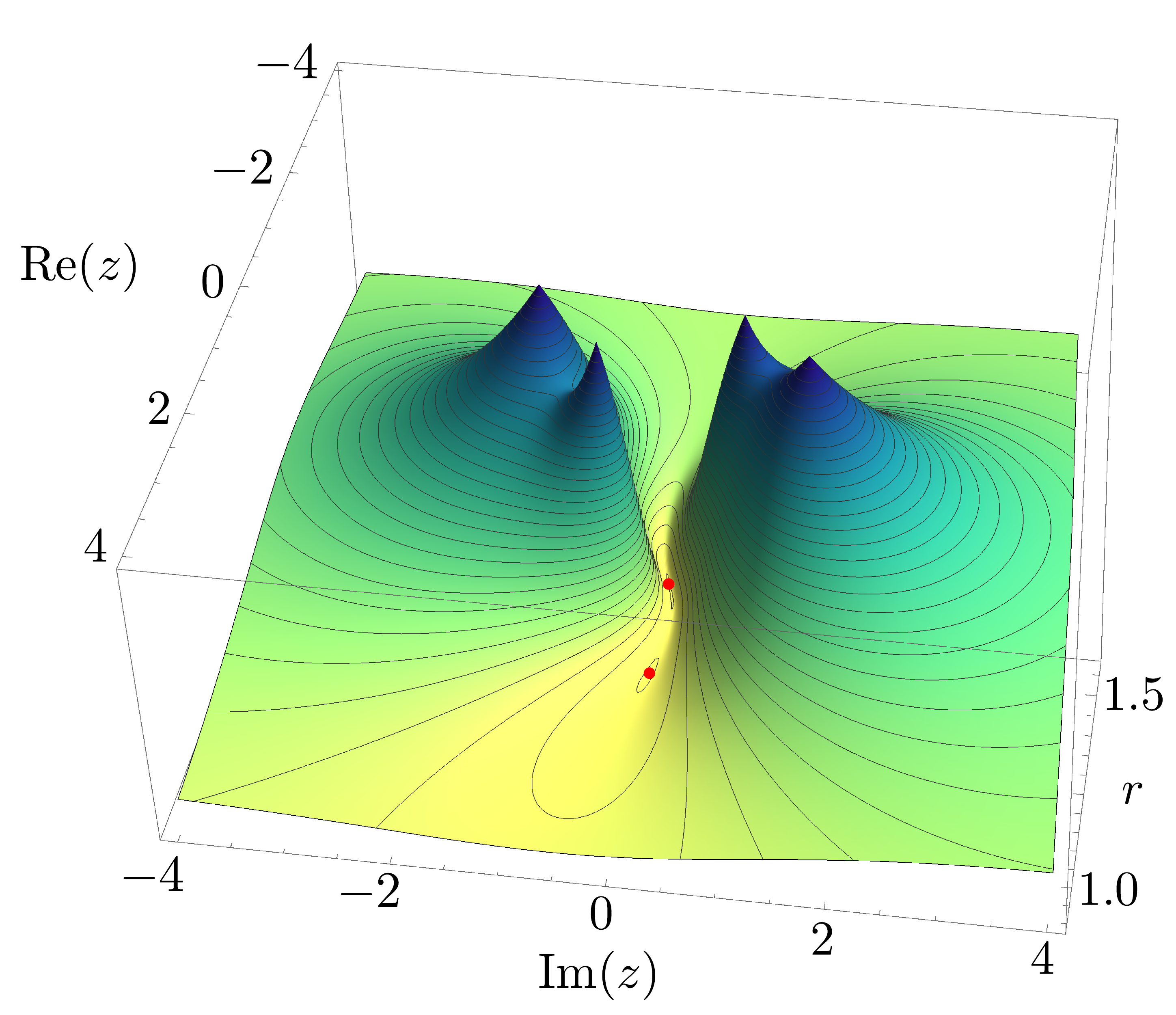}
\caption{\textbf{Top left sphere}: constellation of the spin-2 state 
$\pket {\Psi} \approx (0.634, 0, 0.417 + 0.292 i, 0.053 + 0.048 i,  0.553 + 0.167 i)$, which has two closest SC states. 
\textbf{Top right and middle two spheres:} Plots of $\SSC$, from different viewpoints, with level curves of the Fubini-Study distance to $\pket{\Psi}$. Warmer (online) colors correspond to shorter distances. 
The two red points denote the two closest SC states to $\pket{\Psi}$.
\textbf{Bottom plot:} The above distance function, plotted over the stereographic projection of $\SSC$ on the complex plane. The conical maxima correspond to the directions antipodal to the stars of $\pket{\Psi}$.}
\label{fig:husimiPlot2cpp}
\end{figure}
For $s=1$, there are no states with more than one closest SC state, except for those whose stars are antipodal --- in this latter case the closest SC states form a great circle in the plane that bisects perpendicularly
the diameter connecting the antipodal points. For $s=3/2$ all states with two closest SC states possess a symmetry plane, as is shown below.
For $s=2$ there are states with more than one closest SC states that have no particular symmetry --- an example is shown in figure \ref{fig:husimiPlot2cpp}. 
\begin{theorem}
Let $\pket{\Psi}$ be a spin-$3/2$ state with two closest SC states. Then the constellation associated to $\ket{\Psi}$ is symmetric with respect to
the plane that bisects perpendicularly the segment connecting the stars  of the closest  SC states.
\end{theorem}
\begin{proof}
 Suppose, without loss of generality, that the closest equidistant coherent states point in the directions $n_1=(\theta,\phi=\pi/2)$ and $n_2=(\theta, \phi=3\pi/2)$ --- the bisecting plane is then the $x$-$z$ plane. This implies that
\begin{align*}
\braket{n_1,1/2}{\Psi}
&=0
\, ,
\\
\braket{n_2,1/2}{\Psi}
&=0
\, ,
\\
\braket{n_2}{\Psi}
&=e^{i\gamma}\braket{n_1}{\Psi}
\, ,
\end{align*}
with $\gamma$ a real number. Writing $\ket{\Psi}=(A,B,C,D)$ and imposing
the above conditions leads to
\begin{align*}
A
&=
\lambda(1+3 \cos \theta) \cos(\gamma/2) \,,
\\
B
&=
2\lambda \sqrt{3}\sin (\gamma/ 2) \cos ^2(\theta /2) \cot (\theta/2)
\,,
\\
C
&=
-\sqrt{3}\lambda \cos(\gamma/2) (1+\cos \theta)
\,,
\\
D
&=
-\lambda(1-3\cos \theta) \sin (\gamma/ 2) \cot^3(\theta /2)
\,,
\end{align*}
where $\lambda$, which can be taken real, is fixed by the normalization condition on $\ket{\Psi}$. 
The important point here is that  all components of $\ket{\Psi}$ are real, implying that the coefficients of the corresponding Majorana polynomial are also
real. Therefore, all the  roots of the latter are either real or come in conjugate pairs, so that when projected stereographically  onto the sphere, they give rise to a constellation symmetric with respect to the $x$-$z$ plane, as claimed. 
\end{proof}
\subsection{Conical structure of maxima}
A close look at the maxima of the distance $r$ in figure \ref{fig:husimiPlot2cpp} suggests that the level curves around them are, approximately, circles. We can show that this is the case in general:
 consider a state $\ket{\Psi}=\ket{c_1,\dots,c_N}$ and the corresponding Husimi function defined over $\SSC$. 
Note that $H_\Psi(-c_i)=0$, $i=1,\ldots,N$. We assume, without loss of generality,   that a particular $-c_i$ points toward the north pole. This implies that in the expansion of $\ket{\Psi}$ in $S_z$-eigenstates, the maximal projection eigenstate is absent, $\ket{\Psi}=\sum_{k=-N}^{N-1} \braket{z,k}{\Psi}\ket{z,k}$.
Given any nearby SC state $\ket{n}$, characterized by the angles $(\theta,\phi)$, with $\theta \ll 1$, we compute
\begin{align*}
H_\Psi(n)
&=
|\bra{z}  e^{i \theta S_y}  e^{i \phi S_z} \ket{\Psi}|^2
\\
&=
|\bra{z}(\mathds{1} +i \theta S_y)  e^{i \phi S_z}\ket{\Psi}|^2 +O(\theta)^3
\\
&=
\frac{1}{16} |2\sqrt{2s} \theta e^{i \phi (s-1)}  \braket{z,s-1}{\Psi}|^2+O(\theta^3)
\\
&=
\frac{s}{2}	|\braket{z,s-1}{\Psi}|^2 \theta^2+O(\theta^3) 
\, ,
\end{align*}
where we used the expression of $\ket{\Psi}$ in terms of the eigenstates of $S_z$ to obtain the last line. Since there is no $\phi$ dependence, to this order in $\theta$, we conclude that the blue-colored peaks in Fig.~\ref{fig:husimiPlot2cpp} are, approximately, circular cones. 
\subsection{How do complex lines intersect $S^2_\text{SC}$?}
Another way to explore the way $S^2_\text{SC}$ sits inside the projective space, is to inquire about its intersection with complex lines.   
Theorem~\ref{SCbasis_thm} places severe restrictions in this regard.
\begin{corol}
 For $s \geq 1$, any complex line in $\Ps$ intersects  $\SSC$ at most twice.
\end{corol}
\begin{proof}
Suppose a complex line $\ell$ goes through three SC states $\{ \ket{n_k} \}_{k=1}^3$, then another (non-SC) state  $\ket{\Psi}$ on $\ell$ can be written in the form $\ket{\Psi} = \alpha_1 \ket{n_1} + \alpha_2 \ket{n_2}$ and also $\ket{\Psi} = \beta_1 \ket{n_1} + \beta_2 \ket{n_3}$. Combining the two equations we obtain $  (\alpha_1 - \beta_1 ) \ket{n_1} + \alpha_2 \ket{n_2} - \beta_2 \ket{n_3}= 0$. However, by theorem~\ref{SCbasis_thm}, any 3 SC states are linearly independent for $s\geq 1$, implying that $\ket{n_1} = \ket{\Psi}$, which is a contradiction. 
\end{proof}
In particular, the complex line defined by two SC states  $\ket{n_1}$, $\ket{n_2}$, itself topologically a 2-sphere, only intersects $\SSC$ in these two points. Interestingly, Fermat's (last) theorem for polynomials, a classic result in the  Diophantine inequalities literature~\cite{Lan:90}, is relevant in this regard, as it states that for $A(\zeta)$, $B(\zeta)$, $C(\zeta)$ relatively prime polynomials, the equation
\begin{equation}
\label{FermatPoly}
A(\zeta)^n+B(\zeta)^n=C(\zeta)^n
\, ,
\end{equation}
only has solutions for $n \leq 2$. Taking all three polynomials of the first degree, we deduce that no linear combination of SC states can itself be SC, for $s\geq 3/2$ --- our result above is stronger, as it includes the $s=1$ case.

The following particular case is also of interest: 
\begin{prop}
Given two spin-1 states $\pket{Z}$, $\pket{\mathit{\Xi}}$, with constellations $\{\zeta_1,\zeta_2\}$, $\{\xi_1,\xi_2\}$, respectively, the complex line they define intersects $\SSC$
\begin{enumerate}
\item
in two points, if the states have no star in common
\item 
in the single point $\pket{\chi}$, if the two states have the star $\chi$  in common.
\end{enumerate}
\end{prop}
\begin{proof}
We set a linear combination of the two states equal to an SC state, with associated complex root $\gamma$, 
 which, in terms of Majorana polynomials, implies
\begin{equation}
\alpha_1 (z- \zeta_1)(z- \zeta_2) + \alpha_2 (z- \xi_1)(z- \xi_2) = (z- \gamma)^2
 \, .
\end{equation}
 Solving for $\gamma$, $\alpha_1$, $\alpha_2$,  gives
\begin{align}
\gamma 
&= 
\frac{ \scriptstyle \zeta_1 \zeta_2-\xi_1 \xi_2 \pm \sqrt{(\zeta_1-\xi_1) (\zeta_1-\xi_2) (\zeta_2-\xi_1) (\zeta_2-\xi_2)}}{\scriptstyle \zeta_1+\zeta_2-\xi_1-\xi_2} 
\label{gammasol}
\\
\alpha_1 
&= 
\frac{2 \gamma -\xi _1-\xi _2}{\zeta _1+\zeta _2-\xi _1-\xi _2}
\\
\alpha_2
& = 
\frac{-2 \gamma +\zeta_1+\zeta_2}{\zeta_1+\zeta_2-\xi_1-\xi_2}
 \, .
\end{align}
If the stars of $\pket{Z}$ are different from those of $\pket{\mathit{\Xi}}$,  the radical  in the right hand side of~(\ref{gammasol})  is nonzero, and one obtains two distinct solutions, \ie, the complex line intersects $\SSC$ in two distinct points. 
 On the other hand, if the two states have one star in common, say, $\zeta_1=\xi_1=\chi$, then~(\ref{gammasol}) implies $\gamma=\chi$, \ie, the complex line  intersects $\SSC$ in only  one point, the SC state $\pket{\chi}$ corresponding to the common star. 
 \end{proof}
 Fixing the state $\pket{Z}$ in the previous proposition, and letting $\pket{\mathit{\Xi}}$ range over $\Ps$, one arrives at 
\begin{corol}
 For $s=1$, every complex line through a non SC state $\pket{Z}=\pket{n_1,n_2}$ intersects $\SSC$ twice, except for two lines, each of which intersects $\SSC$  once, at  $\pket{n_i}$, $i=1,2$.
\end{corol}
Another interesting implication is contained in
\begin{corol}
Given a spin-1 state $\ket{\mathit{\Xi}}$, with constellation $\{\xi_1,\xi_2\}$, and an arbitrary SC state $\ket{n}$, with single (multiple) star $\zeta$, $\zeta \neq \xi_1$, $\xi_2$, there exists a unique SC state $\ket{n'}$ such that $\ket{\mathit{\Xi}}$
can be written as a linear combination of $\ket{n}$, $\ket{n'}$. 
\end{corol}
\begin{proof}
Put $\zeta_1=\zeta_2=\zeta$ in~(\ref{gammasol}) to find 
\begin{equation}
\label{gammaz12}
\gamma=\frac{(\xi_1+\xi_2)\zeta-2\xi_1 \xi_2}{2\zeta-(\xi_1+\xi_2)}
\, ,
\end{equation}
\ie, the complex number $\gamma$ corresponding to $n'$ is a M\"obius transform of the one corresponding to $n$, with coefficients that depend on $\ket{\mathit{\Xi}}$. 
\end{proof}
The fact that projective lines, defined by pairs of points in $\SSC$, pass through every point in $\Ps^2$ can be phrased in terms of secant varieties~\cite{Zak:93}: 
the $k$-secant variety $S_k(A,\Ps)$ of a variety $A$ in a projective space $\Ps$ is the (Zariski closure of) the union of all secant $k$-planes  to $A$ (\ie, $k$-planes defined by $k+1$ (non-$k$-coplanar) points of $A$).
\begin{corol}
\label{s1secant}
For  $s=1$, the first secant variety of the spin coherent sphere coincides with the ambient projective space, 
 $S_1(\SSC, \, \Ps^2) = \Ps^2$.
\end{corol}
 For higher values of spin, we have the following
\begin{corol}
Through a point $\pket{\Psi}$ in $\Ps^N$, $N \geq 3$, passes at most one line intersecting $\SSC$ twice.
\end{corol}
\begin{proof}
Assume there are two lines through $\pket{\Psi}$ and intersecting $\SSC$ twice, at $\pket{n_{1}}$, $\pket{n_{2}}$, and $ \pket{m_{1}}$, $\pket{m_{2}}$, respectively.  Then the relation  $\alpha_1 \ket{n_{1}} + \alpha_2 \ket{n_{2}} = \beta_1 \ket{m_{1}} + \beta_2 \ket{m_{2}} $ may  be inferred, and by linear independence of the SC states, $\ket{\Psi} = 0$ follows. 
 \end{proof}
 Note that, as a consequence, for $N \geq 3$, if a state $\ket{\Psi}$ can be written as a linear combination of two SC states, that decomposition is unique.
In $\Ps$, the linear span of two SC states has real dimension at most 6, hence, for $s \geq 2$, there will be states which cannot be expressed as a linear combination of two SC states. For $s=3/2$ such a decomposition is possible, and unique,  for most of the states, as the following proposition asserts
\begin{prop}
\label{N32SC}
 For $s=3/2$, any state $\pket{\Psi}$ without degenerate constellation lies on a complex line defined by two SC states.
\end{prop}
\begin{proof}
Setting $\ket{\Psi}$ equal to a linear combination of the SC states $\ket{n_1}$, $\ket{n_2}$,  implies for the corresponding Majorana polynomials 
\begin{equation}
\label{N3decomp}
(z- \zeta_1)(z- \zeta_2)(z- \zeta_3) = \alpha_1 (z- \gamma_1)^3 + \alpha_2 (z- \gamma_2)^3 
\, .
\end{equation}
Solving for $\gamma_1$, $\gamma_2$, $\alpha_1$, $\alpha_2$, we get
\begin{align}
\label{gammaSolprop1}
\gamma_{1,2} 
&=
A^{-1}
\big(
 \zeta_1^2 ( \zeta_2 + \zeta_3) 
 + 
 \zeta_2^2 (\zeta_3+\zeta_1) 
 + 
 \zeta_3^2 (\zeta_1+ \zeta_2)
 \nonumber
 \\
&
\qquad \qquad
- 6 \zeta_1 \zeta_2 \zeta_3 
 \pm 
 i \sqrt{3}\zeta_{12}\zeta_{23}\zeta_{31} 
 \big)
 \\
 \label{a1Solprop1}
 \alpha_1 
 &= 
 \frac{-3\gamma_2 + \zeta_1 + \zeta_2 + \zeta_3}{3(\gamma_1 - \gamma_2)}
 \\
 \label{a2Solprop1}
 \alpha_2 
 &= 
 \frac{3\gamma_1 - \zeta_1 - \zeta_2 - \zeta_3}{3(\gamma_1 - \gamma_2)}
 \, ,
\end{align}
where 
\begin{equation*}
A \equiv 2
\left(\zeta_1^2 + \zeta_2^2+\zeta_3^2  - \zeta_1 \zeta_2 - \zeta_2 \zeta_3 - \zeta_3 \zeta_1
\right)
\, .
\end{equation*}
\end{proof}
Consider, as an example, the two representative, $s=3/2$, non-biseparable states, $\ket{\text{GHZ}}$ and $\ket{\text{W}}$   \cite{Dur.Vid.Cir:00}. The constellation of the first is a maximal equilateral triangle that can, by a suitable rotation,  be placed on the equator, with one star on the positive $x$-axis. For this orientation, the decomposition in two SC states of proposition~\ref{N32SC} is $\ket{\text{GHZ}}=\frac{1}{\sqrt{2}}(\ket{z}+ \ket{\! - \! z})$. A similar conclusion can be reached from the analysis in~\cite{Gan.Mar.Zyc:12} --- see figure 10 in that reference and the related discussion.

On the other hand, the constellation of the state $\ket{\text{W}}$ consists of two coincident stars, and a third one, antipodal to the other two. As suggested by proposition~\ref{N32SC}, such a state cannot be written as a superposition of two SC states, which is also consistent  with the results of~\cite{Man.Cou.Kel.Mil:14} mentioned earlier. 
Still, it is of interest to inquire what exactly happens if eqs.~(\ref{gammaSolprop1}), (\ref{a1Solprop1}), (\ref{a2Solprop1}), are pushed to their limit in this case. It is easily seen that as $\zeta_3 \rightarrow \zeta_2$ in (\ref{N3decomp}), eq.~(\ref{gammaSolprop1}) implies that $\gamma_1$ and $\gamma_2$ tend to $\zeta_2$, while both $\alpha_1$, $\alpha_2$ blow up. However, a slight reaccommodation of~(\ref{N3decomp}),
\begin{equation}
\begin{split}
(z- \zeta_1)&(z- \zeta_2)(z- \zeta_3) 
\\
& 
= 
(\alpha_1 + \alpha_2) (z- \gamma_1)^3 
\\
&\qquad
+ \alpha_2 \left( (z- \gamma_2)^3 - (z- \gamma_1)^3 \right) 
\end{split}
\, ,
\end{equation}
fixes all problems: the coefficient of the first term on the right hand side is constant, $\alpha_1+\alpha_2=1$, while the exploding $\alpha_2$ in the second term is matched with the vanishing difference $(\zeta-\gamma_2)^3-(\zeta-\gamma_1)^3$, their product having a finite limit,
\begin{align*}
(z- \zeta_1)(z- \zeta_2)^2 
&
= (z- \zeta_2)^3
\\
&\quad
 + \! \! \lim_{\zeta_3 \rightarrow \zeta_2} 
 \! \!
 \alpha_2 \! \left( (\zeta- \gamma_2)^3 \! - \!  (\zeta- \gamma_1)^3 \right) 
 \\
 &=
  (\zeta- \zeta_2)^3 + (\zeta_2 - \zeta_1) (\zeta - \zeta_2)^2 
  \, .
\end{align*}
Clearly, what transpires here is that the spin-3/2 state with a double degeneracy lies on a complex line defined by an SC state and a vector tangent to $\SSC$ at that same state. Thus, states with degenerate constellations are  also  in $S_1(\SSC,\Ps^3)$ and, combining this with proposition~\ref{N32SC} we arrive at a statement analogous to corollary~\ref{s1secant}, for $s=3/2$:
\begin{corol}
\label{s1secantN3}
 $S_1(\SSC, \, \mathbb{P}^3) = \mathbb{P}^3$
 \, .
\end{corol}
We pursue this matter further, studying the case of higher order degeneracies and their relation to tangent varieties, in a forthcoming publication. 

We focus now on the constellations corresponding to the points (states) of a complex line passing through two SC states. Our main result is contained in
\begin{theorem}
\label{the:roots2SC}
Given two  spin-$s$ SC states with roots $\gamma_1$, $\gamma_2 \in \mathbb{C}$, respectively.  The roots $\zeta_k(t)$, $k=0,\ldots,N-1$, of a linear combination of their Majorana polynomials 
\begin{equation*}
 \alpha_1^N (\zeta-\gamma_1)^N- \alpha_2^N  (\zeta-\gamma_2)^N
\end{equation*}
where $\alpha_1=\cos t$, $\alpha_2=e^{i\Omega}\sin t$, trace out circles that intersect equiangularly at $\gamma_1$, $\gamma_2$.
\end{theorem}
\begin{proof}
We compute
\begin{equation}
\begin{split}
 \alpha_1^N (z & - \gamma_1)^N  - \alpha_2^N (z- \gamma_2)^N
 \\
 & = 
\prod_{k=1}^N 
\left( \alpha_1(z- \gamma_1) - \xi^k \alpha_2(z- \gamma_2) \right)
\\
&=
\left( \alpha_1^N - \alpha_2^N \right)\prod_{k=1}^N
 \left( z - \frac{ \gamma_1 \alpha_1  - \gamma_2 \xi^k \alpha_2 }{\alpha_1 - \xi^k \alpha_2  } \right)
\end{split}
 \, ,
\end{equation}
with $\xi=e^{i 2 \pi /N}$ a primitive $N$th root of  unity, which shows that
\begin{equation}
\label{root.comb}
\zeta_k (t) 
= 
\frac{ \gamma_1 \cos t  - \gamma_2 e^{i \Omega} \xi^k \sin t }{\cos t - e^{i \Omega} \xi^k \sin t }
\, .
\end{equation}
Consider now the M\"obius transformation $M(\zeta)= (\zeta-\gamma_1)/(\zeta - \gamma_2)$
and substitute from (\ref{root.comb}) to find
\begin{equation}
M(\zeta_k (t) ) = e^{i (\Omega+2\pi k/N)}  \tan t 
\, ,
\end{equation} 
which is a line through the origin making an angle $\Omega + 2 \pi k /N$ with the real axis. The proof is completed by noting that M\"obius transformations are conformal.
\end{proof}
\noindent 
 We make some related comments: 
 \begin{enumerate}
 \item
 The theorem could be stated in terms of a linear combination of the states themselves --- passing to the corresponding Majorana polynomials involves a rescaling of the coefficients in the linear combination. 
 \item
As usual, ``circles'', in the complex plane, includes the case of straight lines through the origin (see, \eg, top of figure~\ref{zetakorbits}).
 \item
Given that stereographic projection is also conformal, we may conclude that the trajectories of the stars on the Bloch sphere are also circles, intersecting equiangularly. This fact, for the case $s=1$, has been pointed out before --- see figure 11 in~\cite{Gan.Mar.Zyc:12}. 
\item
The theorem provides a  proof of the fact that a superposition of two SC states  cannot produce a state with degenerate stars, as suggested, for $s=3/2$, in proposition~\ref{N32SC}. 
\end{enumerate}
A particular $s=3/2$ case is depicted in figure~\ref{zetakorbits}.
\begin{figure}
\includegraphics[scale=0.3]{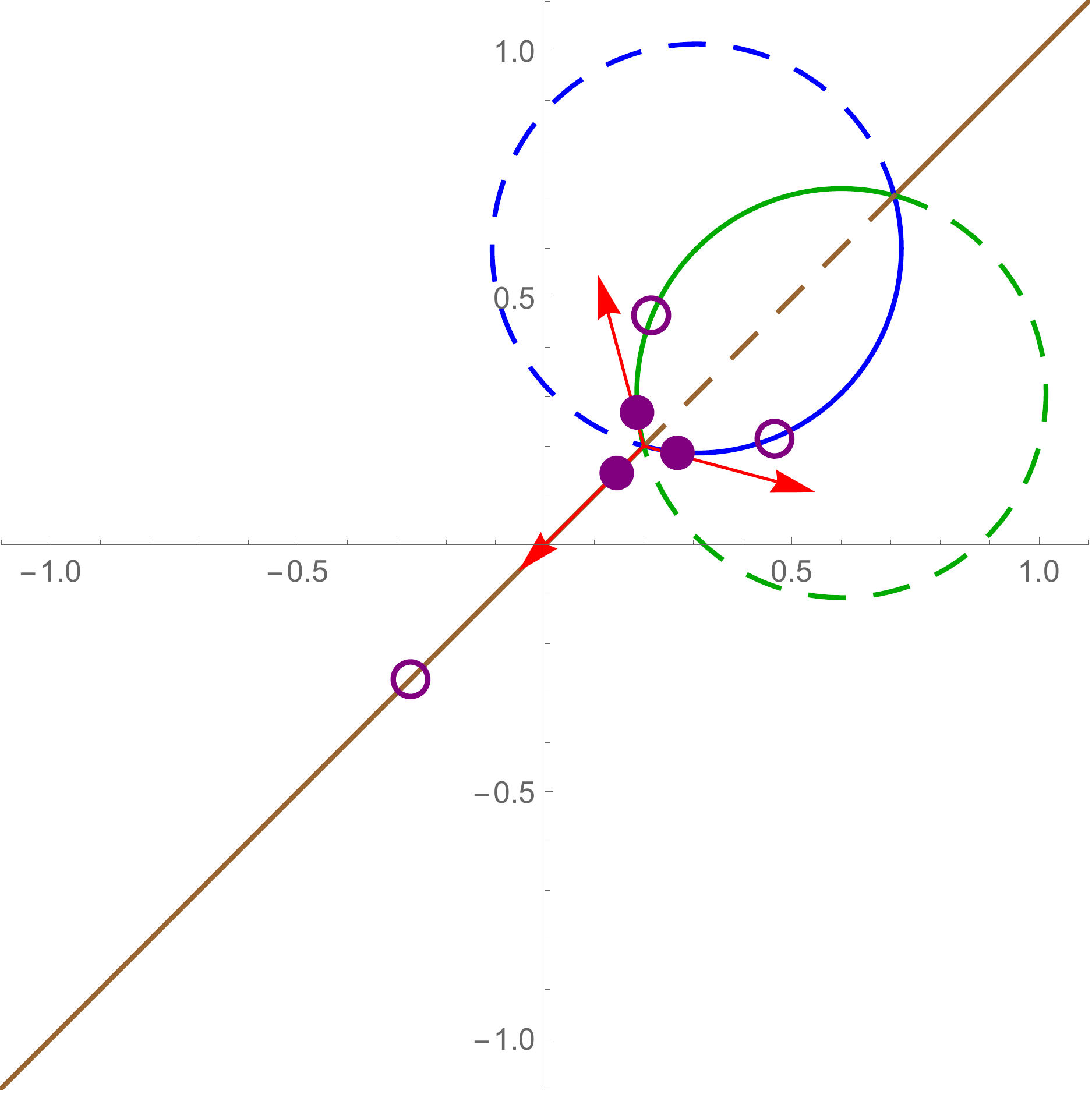}
\\
\includegraphics[scale=0.35]{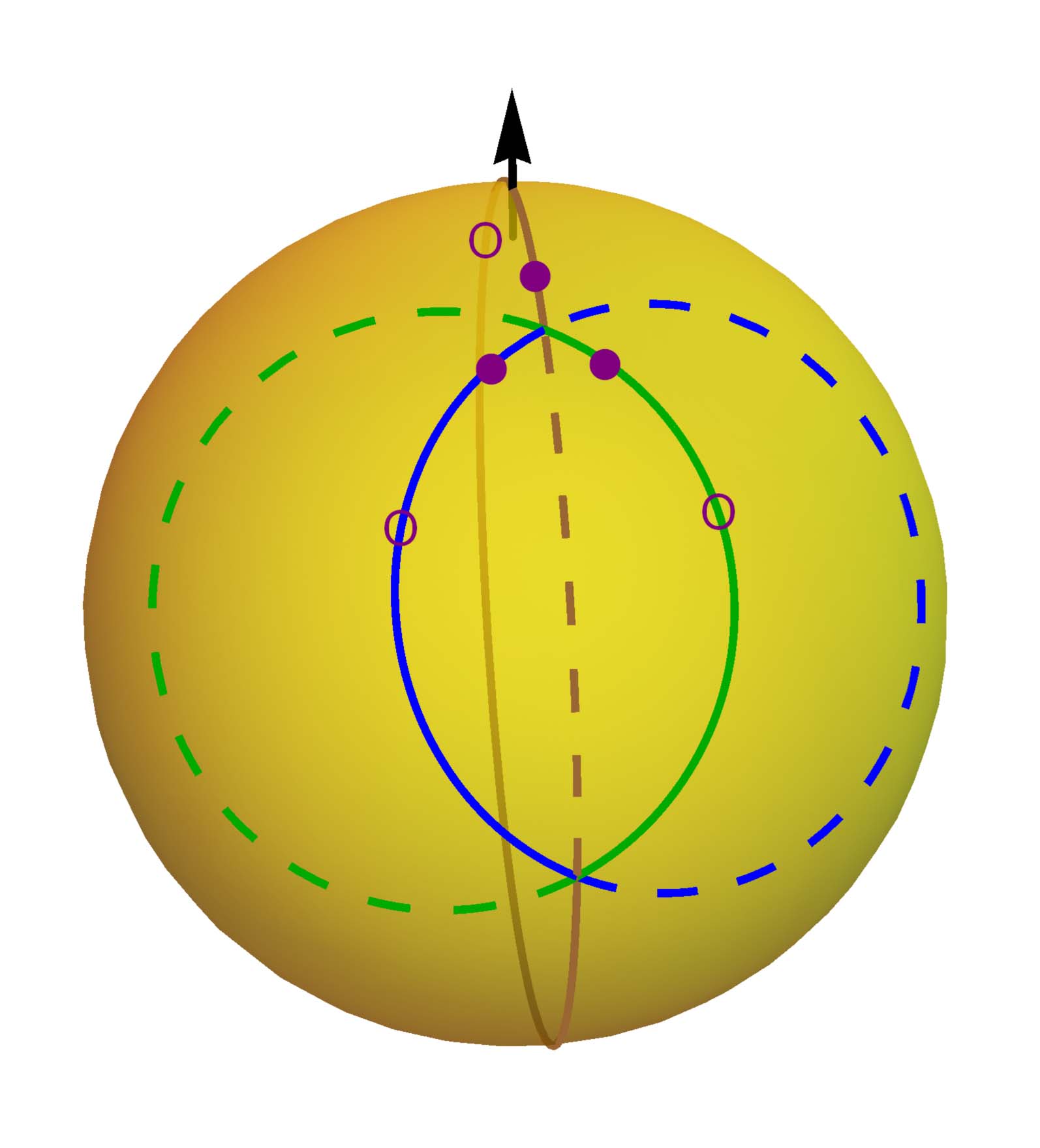}
\caption{%
Plot of the curves $\zeta_k(t)$ in~(\ref{root.comb}) in the complex plane (top) and its stereographic projection (bottom), for $s=3/2$, $\gamma_1= (1+i)/5$, $\gamma_2=(1+i)/\sqrt{2}$ and $\Omega=0$. The solid part of each circle takes a root $\zeta_i$ from $\gamma_1$ to $\gamma_2$, for $0 \leq t \leq \pi/2$, while the dashed part returns it from $\gamma_2$ to $\gamma_1$, for $\pi/2 \leq t \leq \pi$. The red arrows at $\gamma_1$, in the top figure,  are the vectors tangent to the curves at $t=0$, with an angle $2\pi/3$ between any two of them. The solid dots and little circles denote the configuration of the constellation at $t=0.1$ and $t=0.45$, respectively.
}
\label{zetakorbits}
\end{figure}

We end this subsection with a general statement about the number of distinct stars of a linear combination of \emph{any} two  states. To begin with, note that if the states share a star $n$, with multiplicities, say, $r$, $s$, respectively, then a linear combination of them will also have $n$ as a star, with multiplicity equal to $\min(r,s)$. Clearly, an analogous result holds in the case of several stars $\{n_i\}$ in common, each with different multiplicities $\{r_i\}$, $\{s_i\}$, in the two states. When factoring the linear combination of the two corresponding Majorana polynomials, such common factors may be canceled, and the problem reduces to that of a lower spin, without common stars. Therefore, we may assume, without loss of generality, that the states in question have no stars in common (note though that each state may have stars with multiplicity). Then the following result holds
\begin{theorem}
\label{cor:roots2SC}
Consider two  spin-$s$ states, $\ket{\Psi_1}$, $\ket{\Psi_2}$, with $n_1$, $n_2$ distinct stars respectively (each with possible multiplicity),  of which none are in common between the two states. Then an arbitrary linear combination $\ket{\Phi}=a\ket{\Psi_1}+b \ket{\Psi_2}$ has itself at least $N -n_1-n_2+1$  distinct stars (each with possible multiplicity).  
\end{theorem}
\begin{proof}
The statement is an immediate consequence of Mason's theorem~\cite{Mas:83,Mas:84,Lan:90}. Let $n_0(F(\zeta))$ denote the number of distinct roots of the complex polynomial $F(\zeta)$. Let $A$, $B$, $C$ be relatively prime polynomials such that $A+B=C$. Then Mason's theorem states that 
\begin{equation}
\label{Masonineq}
\max \deg\{A,B,C\} \leq n_0(ABC)-1
\, .
\end{equation}
To apply this to our case, put 
\begin{equation*}
A=a p_{\Psi_1}
\, ,
\quad
B=b p_{\Psi_2}
\, ,
\quad
C=p_{\Phi}
\, ,
\end{equation*}
with the Majorana polynomials as in~(\ref{pjk}), so that $n_0(A)=n_1$, $n_0(B)=n_2$, and, say, $n_0(C)=n_3$.
With our assumption about no common roots,  if one of the $\ket{\Psi_i}$ has a star at the south pole, and, hence, the degree of its Majorana polynomial is less than $N$, the other cannot also have a star there, and the left hand side of~(\ref{Masonineq}) is, in all cases, equal to $N$. Note also that if $C$ shared a root with, say, $A$, then it would have to also share it with $B$, which contradicts our assumptions, so all $n_3$ distinct roots of $C$ are different  from those of $A$ and $B$. Then the number of distinct roots of the product $ABC$ is $n_0(ABC)=n_1+n_2+n_3$, and the statement follows from~(\ref{Masonineq}).
\end{proof}
For the case of two SC states, $n_1=n_2=1$, we get $n_3 \geq N-1$, which is weaker than our result that in fact $n_3=N$. On the other hand, for two states with star multiplicities, such that $n_1+n_2<N$, we get $n_3 \geq 2$, which is a new result: the complex line through such states does not intersect $\SSC$.
\subsection{Visualizing $S^2_{\textbf{SC}}$}
The motivation for this subsection came from our struggling with the mental picture we presented at the beginning of  section~\ref{CSCs}: an expanding geodesic sphere $S_r$ that ends up tangent to $\SSC$ at exactly $N$ points. Now, intersections of submanifolds are robust --- wiggling a little bit the intersecting parts one still ends up with an intersection, but tangencies are not: when perturbed, they either disappear, or get converted to intersections. It is a bit puzzling then that the above two spheres remain tangent at $N$ points, for \emph{any} position of the center $\pket{\Psi}$  of $S_{\pi/2}$ (the $N$ points of tangency, of course, change, as $\pket{\Psi}$ is moved around in $\Ps$). Looking at figure~\ref{fig:husimiPlot2cpp}, and trying to imagine the surface depicted there wrapped around $\SSC$, we arrive at the cartoon in figure~\ref{fig:cartoon}, where, for simplicity, we have assumed that $s=1$, so that there are only two ``peaks'' on $\SSC$. 
%
%
%
\begin{figure}
\includegraphics[scale=0.5]{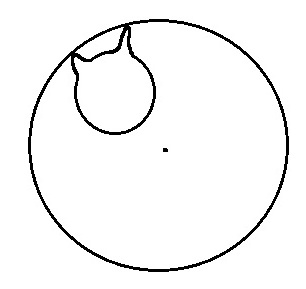}
\caption{Artist's rendition of the geodesic sphere $S_{\pi/2}$ (outlying circle), centered at $\pket{\Psi}$, and of $S^2_\text{SC}$ (cat shaped curve), tangent to $S_{\pi/2}$ at two points (assuming $s=1$).}
\label{fig:cartoon}
\end{figure}
But this image is hardly convincing: for example, how  are the peaks compatible with the known fact that the restriction of the Fubini-Study metric on $\SSC$ gives a perfectly ``round'' sphere, with constant curvature? And how can $S_{\pi/2}$ remain tangent to $\SSC$ when $\pket{\Psi}$ is moved freely in $\Ps$? Worse still,  how many peaks does $\SSC$ really have, if any?  Now, some of these puzzles are simply byproducts of vague phrasing, naively drawn cartoons, and other such easily fixable looseness. For example, $S_{\pi/2}$ in figure~\ref{fig:cartoon} is actually a codimension-1 object (\eg, 3D for $s=1$), which is certainly not what that image conveys. Other aspects of these questions though seem to persist, even when elementary corrections are taken into account. We felt, therefore,  that a good starting point in trying to answer them would be ``taking a picture'' of $\SSC$, from $\pket{\Psi}$'s position. In this, we assume that the light used to take the picture follows Fubini-Study geodesics, and use the inverse of the exponential map, based at $\pket{\Psi}$, to lift the image of $\SSC$ into the tangent space at $\pket{\Psi}$ --- the result is what we called $\log_\Psi(\SSC)$, and we wonder what it looks like. Our theorem~\ref{thm:maxdim} guarantees we can only  plot projections of the 4D image in, say, 3-planes, and that is indeed our goal. We sketch the calculation, fixing, for simplicity, $s=1$, and identifying a point $\pket{\Psi}$ with the density matrix $\rho_\Psi=\ket{\Psi}\bra{\Psi}$. The two stars $n_1$, $n_2$, of $\ket{\Psi}$ are taken in the $x$-$z$ plane, symmetrically with respect to the $z$-axis, and making an angle $\alpha \in [0,\pi/2]$ with it, \ie,
\begin{align}
\label{n1n2def}
n_1
&=
(\sin\alpha,0,\cos\alpha)
\\
n_2
&=
(-\sin\alpha,0,\cos\alpha)
\, .
\end{align}
 The corresponding state, in the $S_z$-basis $(1,0,-1)$, is
\begin{equation}
\label{psialphadef}
\ket{\Psi}
=
2b^{-1}
\left(
\cos^2 \frac{\alpha}{2}, 0, -\sin ^2 \frac{\alpha}{2}
\right)
\, ,
\end{equation}
with $b \equiv \sqrt{3+\cos 2\alpha}$. The SC states corresponding to the stars are
\begin{align}
\label{n1Sdef}
\ket{n_1}
&=
\left(
\cos^2 \frac{\alpha}{2}, \frac{1}{\sqrt{2}} \sin\alpha, \sin^2 \frac{\alpha}{2}
\right)
\\
\label{n2Sdef}
\ket{n_2}
&=
\left(
\cos^2 \frac{\alpha}{2},- \frac{1}{\sqrt{2}} \sin\alpha, \sin^2 \frac{\alpha}{2}
\right)
\, .
\end{align}
The curve
\begin{equation}
\label{ctcurve}
\ket{c(t)}
=
(\cos t- \cot \omega \sin t)\ket{\Psi}
+e^{i \eta} \sin t \csc \omega \ket{n}
\, ,
\end{equation}
in $\Hs$, where $\braket{n}{\Psi} \equiv \cos \omega \, e^{i \eta}$, projects to a geodesic $\rho_{c(t)}$ in $\Ps$, starting, at $t=0$, at $\rho_{\Psi}$ and reaching, at $t=\omega$, the SC state $\rho_{n}$. The tangent vector $\partial_t \rho_{c(t)}|_{t=0} \equiv \dot{\rho}_c(0)$ is given by
\begin{equation}
\label{rhodt0}
 \dot{\rho}_c(0)
 =
 -2 \cot \omega \rho_\Psi
 +
 \csc \omega 
 \left(
 e^{i \eta} \ket{n}\bra{\Psi}
 +
 e^{-i\eta} \ket{\Psi}\bra{n}
 \right)
 \, ,
 \end{equation}
 and is of unit length, as $t$ is arclength along $\rho_c(t)$. Then 
 \begin{equation}
 \label{logrhon}
 v_n 
\equiv 
\log_\Psi  \rho_n
=
\omega \dot{\rho}_c(0)
\, ,
\end{equation}
is the sought image of $\rho_n$ in $T_\Psi \Ps$, since $\omega$ is the geodesic distance between $\rho_\Psi$ and $\rho_n$. We choose an orthonormal hermitian basis 
$
\{ 
h_1, h_2, h_3, h_4
\}
$ in $T_\Psi \Ps$, where
\begin{align}
\label{hidef}
H_1 
&\equiv 
h_1 + i h_2
\\
&=
2b^{-1}
\left(
\begin{array}{ccc}
0 & 0 & 0
\\
\scriptstyle
\cos \alpha+1 & 0 &
\scriptstyle
\cos\alpha - 1
\\
0 & 0 & 0
\end{array}
\right)
\\
H_2
&\equiv
h_3 + i h_4
\\
&=
b^{-2}
\left(
\begin{array}{ccc}
\scriptstyle
1-\cos 2\alpha
& 
0
&
\scriptstyle
-8 \sin^4 \frac{\alpha}{2}
\\
0 & 0 & 0
\\
\scriptstyle
3+\cos 2\alpha +4 \cos \alpha
&
0
&
\scriptstyle
\cos 2\alpha -1
\end{array}
\right)
\, ,
\end{align}
and compute the corresponding components $v_n^i=\frac{1}{2}\Tr(v_n h_i)$,
\begin{align}
v_n^1+i v_n^2
&=
\frac{\sqrt{2} e^{-i \phi}  \omega \chi \sin \theta}{b \sin \omega \cos \omega}
\\
v_n^3+i v_n^4
&=
\frac{4 e^{-i 2\phi} \omega \chi \xi}{b^2 \cos \omega \sin \omega}
\, ,
\end{align}
where
\begin{align}
\chi
&=
\cos^2 \frac{\alpha}{2} \cos^2 \frac{\theta}{2}-e^{i 2\phi} \sin^2 \frac{\alpha}{2} \sin^2 \frac{\theta}{2}
\\
\xi
&=
\cos^2 \frac{\alpha}{2} \sin^2 \frac{\theta}{2} +e^{i 2\phi} \sin^2 \frac{\alpha}{2} \cos^2 \frac{\theta}{2}
\, ,
\end{align}
and the phase of the SC states was chosen so that $\braket{n}{\Psi}=\cos \omega \geq 0$.
We plot the projection of $\log_\Psi \SSC$ in the 123-plane, for $\alpha=\pi/12$, $\pi/3$, and $\pi/2$, in figure~\ref{fig:S2SC_1}. Since normal coordinates, centered at $\rho_\Psi$, are being used, $\rho_\Psi$ lies at the origin in the figure and  Fubini-Study geodesics through it look like straight lines.
\begin{figure*}
\centering
\raisebox{4ex}{\includegraphics[scale=0.25]{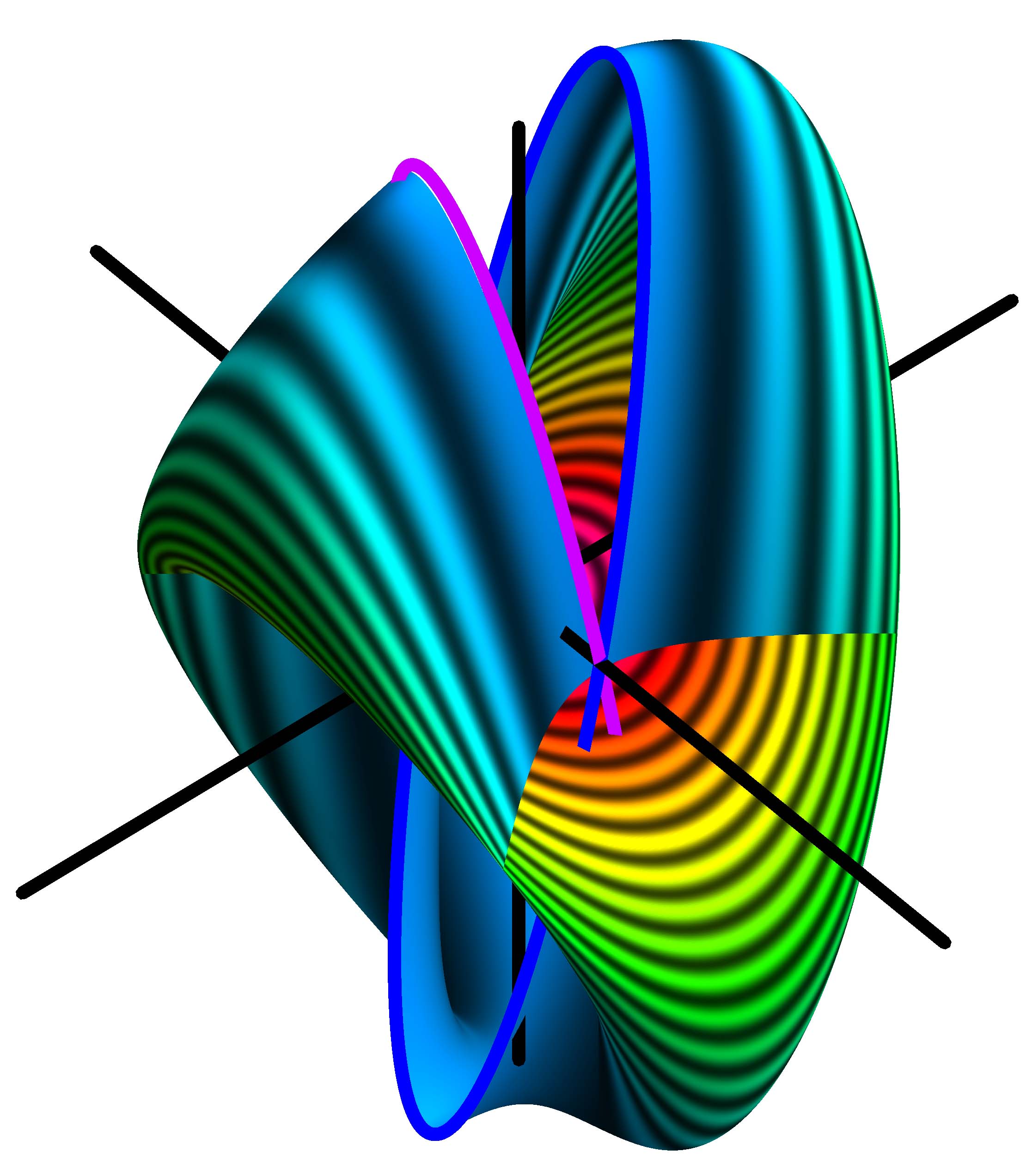}}
\hfill
\raisebox{1.4ex}{\includegraphics[scale=0.33]{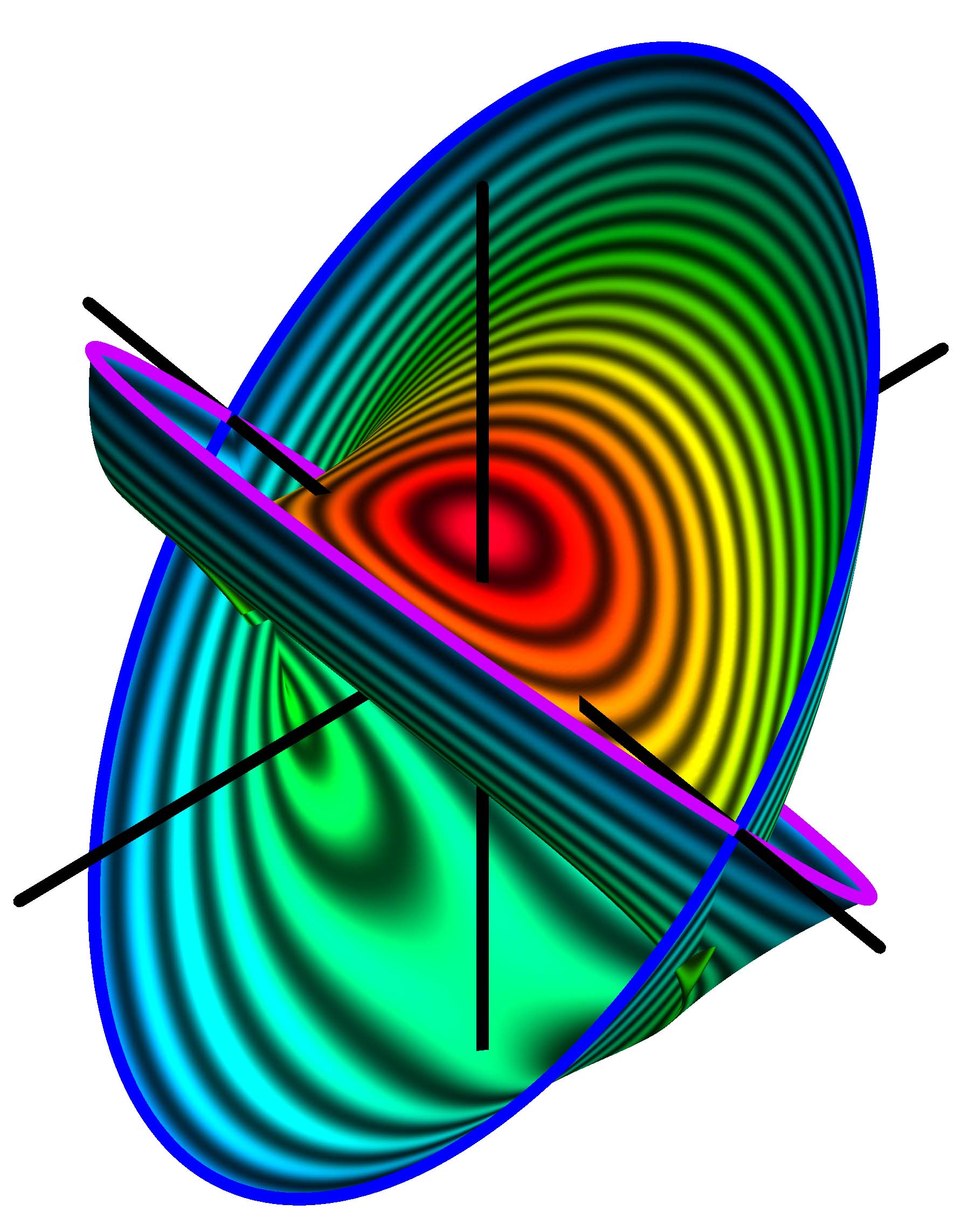}}
\hfill
\raisebox{1.5ex}{\includegraphics[scale=0.25]{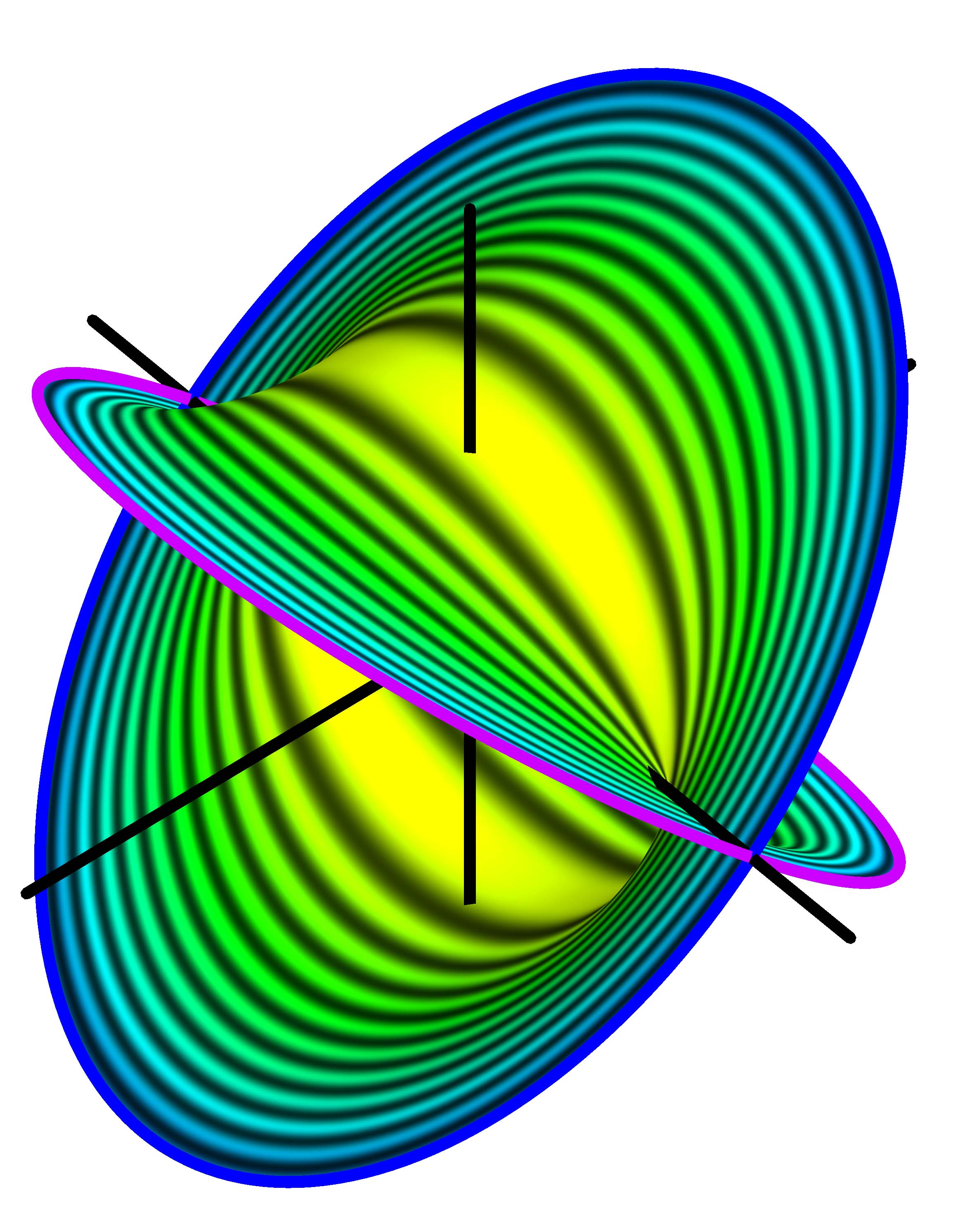}}
\caption{%
Plot of $\log_\Psi(\SSC)$, with $\ket{\Psi}$ as in~(\ref{psialphadef}), for $\alpha=\pi/12$ (left), 
$\pi/3$ (center), $\pi/2$ (right) (projection in the plane 123). The state $\pket{\Psi}$ is at the origin, where the axes intersect (not visible). The highlighted ellipses are the inverse images, under $\exp_\Psi$, of the  SC states $\pket{-n_1}$, $\pket{-n_2}$ in directions antipodal to the stars of $\pket{\Psi}$ --- the singularities of $\log_\Psi$ there blow up individual points to entire circles. The (online) color coding  assigns warmer colors  to the SC states closest to $\pket{\Psi}$ (red for the north pole of $S^2_\text{SC}$ in the first two plots, yellow for the 23 meridian in the third plot), and blue to those farthest away (above mentioned ellipses). The rapid brightness modulation marks equidistance from $\pket{\Psi}$ --- note how it slows down near the above mentioned extrema.
}
\label{fig:S2SC_1}
\end{figure*}
A notable, and initially puzzling, feature of the surface shown in that figure, supposedly the image of a topological 2-sphere, is that it seems to have a boundary: one sees a self-intersecting surface that  ends on two ellipses (highlighted in blue/violet). The latter are the projections, in the 123-plane, of two circles in the full, 4D tangent space. In their turn, the circles are the inverse images, under the exponential map, of the SC states in the directions antipodal to the stars of $\rho_\Psi$. What happens here is that $\ket{\Psi}=\ket{n_1,n_2}$ is orthogonal to $\ket{\! - \! n_i}$, $i=1,2$, so that $\pket{\Psi}$ and, say, $\pket{-n_1}$, are antipodal points on the projective line (real 2-sphere) they define. Then $\pket{- n_1}$ is in the cut locus of $\exp_\Psi$ and all vectors tangent to the above 2-sphere at $\pket{\Psi}$, of length $\pi/2$, ``point'' to $\pket{- n_1}$ --- the circles (ellipses) in the figure are just the loci of those tangent vectors. Going up one dimension, in the full tangent space, the geodesic sphere $S_{\pi/2}$ would look like a euclidean 3-sphere centered at the origin, where $\pket{\Psi}$ lies, and the above circles are great circles on that sphere. This last statement of course needs to be taken with a grain of salt, as $S_{\pi/2}$ is in its entirety in the cut locus of $\exp_\Psi$, but it can be made  precise in a limiting sense.

Two further snapshots of $\SSC$ for $\alpha=\pi/3$, from different viewpoints,  are shown in figure~\ref{fig:S2SC_2} (left and middle plots). In the one in the middle, the complex line defined by $\pket{n_1}$, $\pket{n_2}$, is also plotted --- rather than a topological 2-sphere, it looks like a spherical cap, the reason being that the state $\sqrt{2/3}(\ket{n_1}-\ket{n_2})$, which belongs to that complex line, is orthogonal to $\ket{\Psi}$, so its logarithm is, as we have seen above, an entire circle (the boundary of the cap). Note that this is the rule, rather than the exception: any generic complex line $\ket{\phi_1}+\zeta \ket{\phi_2}$ contains a single state $\ket{\Psi}^\perp$, orthogonal to a given state $\ket{\Psi}$, corresponding to $\zeta=-\braket{\Psi}{\phi_1}/\braket{\Psi}{\phi_2}$. That state will blow up into a full circle under $\log_\Psi$, and, accordingly, the complex line, rather than a 2-sphere, will look like a cap, with $\log_\Psi(\pket{\Psi}^\perp)$ at its boundary. 

Another way to visualize $\SSC$ is to use $\dot{\rho}_c(0)$, in equation~(\ref{rhodt0}), to map $\SSC$ to a surface in the unit tangent sphere $S^3$ at $\pket{\Psi}$. Thus, the radial information about $\SSC$ is erased, and the above mentioned surface only records the direction in which each point of $\SSC$ is viewed from $\pket{\Psi}$. That surface, in turn, may be stereographically projected, from the ``south'' 4-pole to the 123-equatorial plane in $T_\Psi \Ps$ --- the result is plotted in the right in figure~\ref{fig:S2SC_2}. Note that the two circles that correspond to the SC states $\pket{-n_1}$, $\pket{-n_2}$ are linked.
\begin{figure*}
\centering
\raisebox{1.4ex}{\includegraphics[scale=0.25]{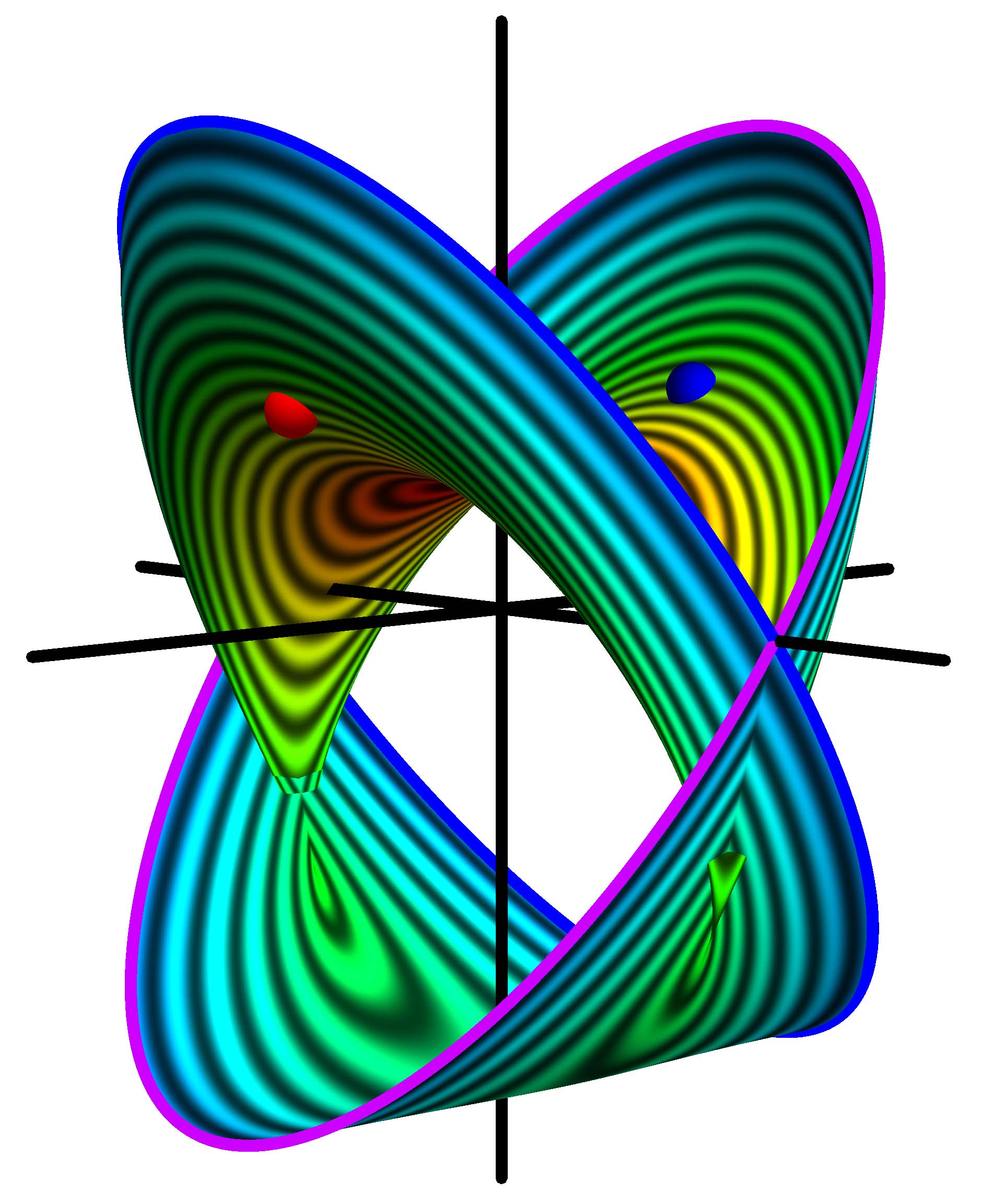}}
\hfill
\raisebox{1.5ex}{\includegraphics[scale=0.25]{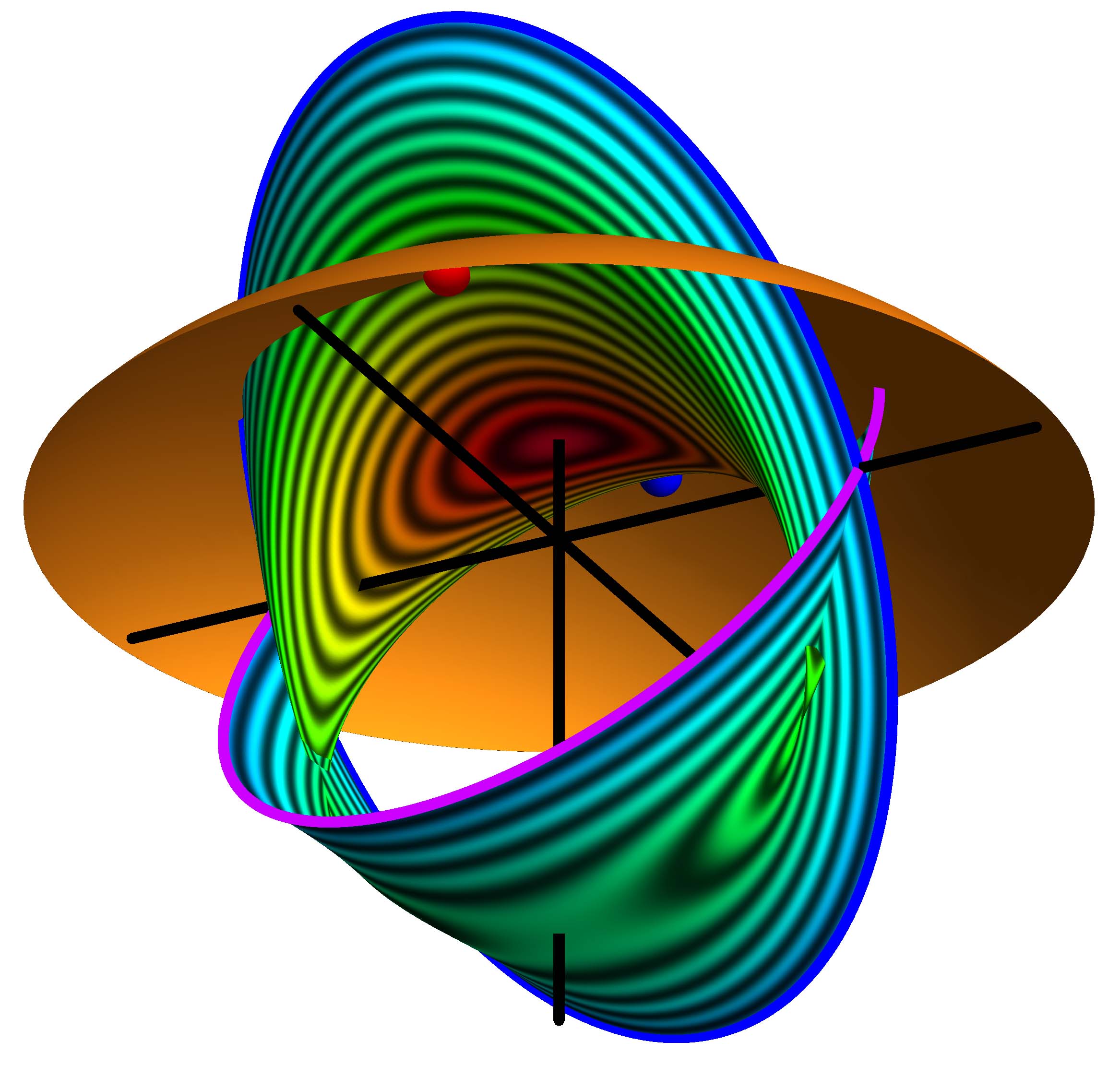}}
\hfill
\raisebox{4.5ex}{\includegraphics[scale=0.25]{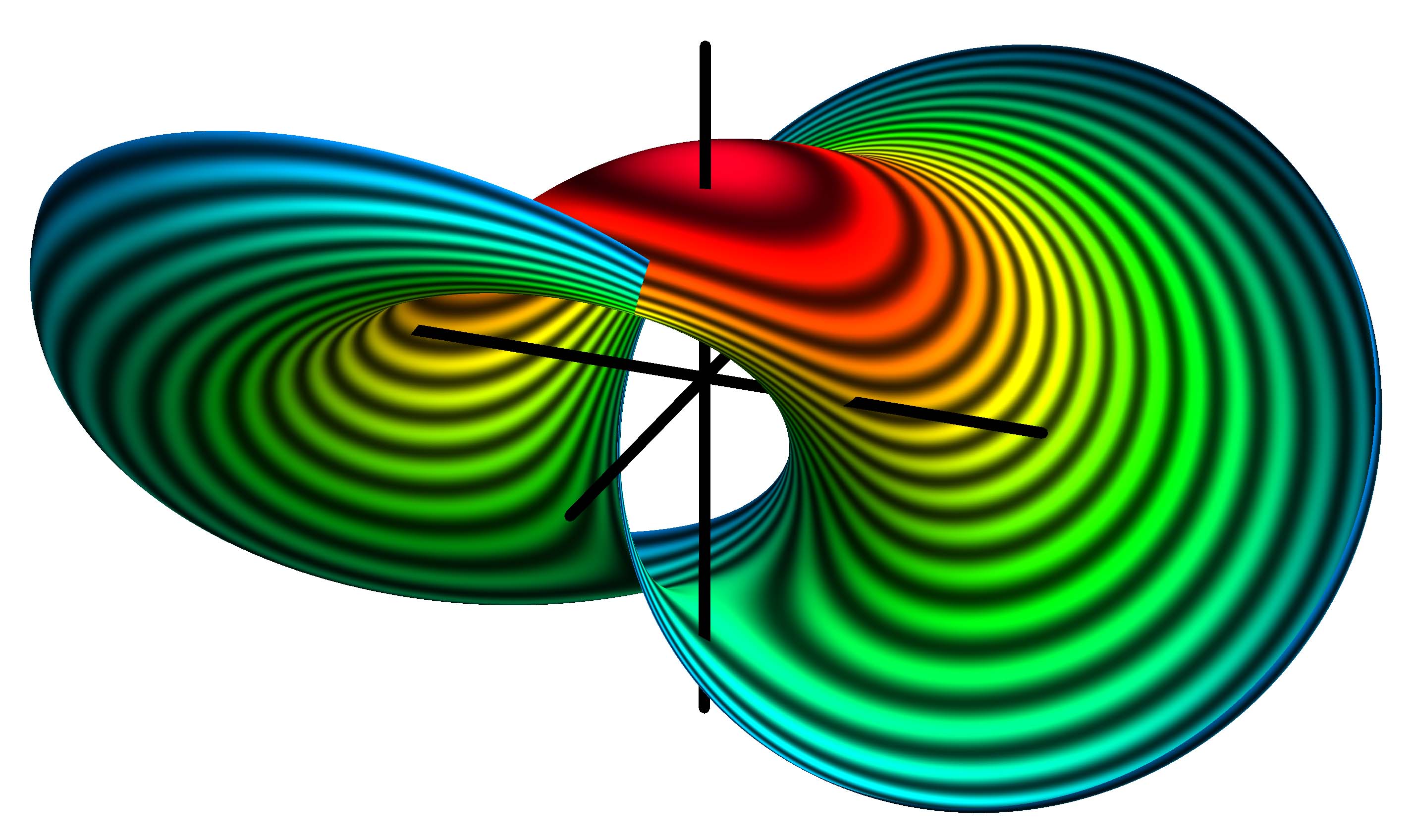}}
\caption{%
\textbf{Left and middle:} Shown is the surface in the middle of Fig.~\ref{fig:S2SC_1} ($\log_\Psi(\SSC)$ for $\alpha=\pi/3$), from two different viewpoints. The two little spheres on the surface denote the position of the SC states $\pket{n_1}$, $\pket{n_2}$ (see~(\ref{n1Sdef}), (\ref{n2Sdef})), corresponding to the stars of  $\pket{\Psi}$ ($\pket{\Psi}$ itself is at the origin). The ``spherical cap'' superimposed in the figure in the middle is the complex line $\ell$ passing through $\pket{n_1}$, $\pket{n_2}$ --- although topologically a 2-sphere, it appears to have a boundary, because the state  $\ket{\Psi}^\perp=\sqrt{2/3}(\ket{n_1}-\ket{n_2})$, which  belongs to $\ell$ and is orthogonal to $\ket{\Psi}$, is blown up into a circle (the boundary of the cap) under $\log_\Psi$.
 Note that (the projection of) $\SSC$, rather than a moon-like object in the horizon, appears to ``wrap around the sky'', when viewed from $\pket{\Psi}$. Note also that, in the full (4D) $T_\Psi \Ps^2$, $\ell$ only intersects $\SSC$ in the two points $\pket{n_1}$, $\pket{n_2}$ --- additional intersections appearing in the figure are an artifact of the projection in the 123-hyperplane.
 \textbf{Right:} Stereographic projection from the south 4-pole  to the equatorial 123-hyperplane in $T_\Psi \Ps$ of the image of $\SSC$ under the map $\dot{\rho}_c(0) \colon \SSC \rightarrow S^3 \subset T_\Psi \Ps$ in equation~(\ref{rhodt0}). The (online)  color coding in all three plots is as in figure~\ref{fig:S2SC_1}.
}
\label{fig:S2SC_2}
\end{figure*}

A further interesting result can be inferred from~(\ref{rhodt0}). To begin with, that relation is valid with $\ket{n}$ being replaced by a general (\ie, not necessarily SC) state $\ket{a}$. We use the notation $\braket{a}{\Psi}=\cos \omega_{a\Psi} e^{i \eta_{a\Psi}}$ for any pair of states. Call $v_a$ the unit vector tangent at $\pket{\Psi}$, pointing towards $\pket{a}$, and similarly for $v_b$. Then, the angle $\Theta_{ab}$ between $v_a$, $v_b$, is found to be 
\begin{align}
\cos  \Theta_{ab}   
&=
\frac{1}{2} \Tr(v_a v_b)
\nonumber
\\
&=
\frac{\cos \omega_{ab} \cos\Omega-\cos \omega_{a\Psi} \cos \omega_{b\Psi} }{\sin \omega_{a\Psi} \sin \omega_{b\Psi}}
\, ,
\label{ProjTri}
\end{align}
where $\Omega=\eta_{ab}+\eta_{b\Psi}+\eta_{\Psi a}$ is the phase of the Bargmann invariant of the three states involved,
\begin{equation}
\label{Binv}
\braket{a}{b}\braket{b}{\Psi}\braket{\Psi}{a}=Re^{i \Omega}
\, ,
\end{equation}
where $R,\Omega \in \mathbb{R}$. Note that, for $\Omega=0$, (\ref{ProjTri}) reduces to the formula for the angle of a spherical geodesic triangle in terms of the lengths (angles) of its sides. This is not an accident, in fact~(\ref{ProjTri}) \emph{is} the spherical trigonometric formula, only expressed in terms of projective space quantities. To see this, consider the real version of the Hilbert space $\Hs$, with $\Hs \ni \ket{\Psi}=(x_0+i y_0,\ldots ,x_N+i y_N) \rightarrow (x_0, \ldots, x_N,y_0, \ldots, y_N) = \Psi \in \mathbb{R}^{2N+2}$, so that normalized kets in $\Hs$ are mapped to the unit sphere $S^{2N+1}$ in $\mathbb{R}^{2N+2}$. The euclidean inner product between two such vectors $\Psi$, $\Phi$, is easily seen to be given by $\Psi \cdot \Phi=\Re \braket{\Psi}{\Phi}$, so that the angle $s$ between them satisfies
\begin{equation}
\label{sreta}
\cos s= \Psi \cdot \Phi = \Re \braket{\Psi}{\Phi}= \cos \omega_{\Psi \Phi} \cos \eta
\, ,
\end{equation}
where $\braket{\Psi}{\Phi}=\cos \omega_{\Psi \Phi} e^{i \eta}$, and $\omega_{\Psi \Phi}$ is the FS distance between $\pket{\Psi}$, $\pket{\Phi}$ in $\Ps$. When the two states are in phase, \ie, their inner product is positive, their distance on $S^{2N+1}$ is equal to the FS one of their images in $\Ps$ --- (\ref{ProjTri}) then follows, keeping in mind that the SC states where assumed in phase with $\ket{\Psi}$.
%
%
%
%
%
%
%
%
%

\section{Summary and Concluding Remarks}

%
%
%
%
%
We have investigated questions regarding the intersection of complex lines and Fubini-Study geodesics in quantum projective state space $\Ps$
 with the 2-sphere of spin coherent states $\SSC$ --- a central role in this discussion is played by our result of the linear independence of any $N+1$ SC states. We showed that for a generic quantum state $\pket{\Psi}$, there exists an adapted SC basis, defined via  the extrema of its Husimi function. We also gave a lower bound on the number of distinct stars of a linear combination of two generic spin-$s$ states, and found a simple expression for the constellation of a linear combination of two spin-$s$ SC states. Finally, we computed the image of the SC 2-sphere, for $s=1$, projected to a 3D subspace of the tangent space to $\Ps^2$, using (the inverse of) the exponential map.
 
 As mentioned before, our motivation in delving into this sort of questions, of a distinctly algebraic geometric flavor, is mainly rooted in our belief that the answers naturally translate into statements that an experimentalist might find not only neat but also useful. Our initial excursion into this territory has left many stones unturned. A basic piece of information that seems missing is the form of the Majorana constellation obtained by linearly combining two given states. This leads back to the mostly open problem of factorizing a sum of polynomials, but apart from a complete description of the result, which might be presently  untenable, one may also envisage partial answers in terms of  bounds and inequalities, already unearthed but hidden deeply in the mathematics literature. Another promising direction seems to be ``intersectology'', hopefully streamlined by a more substantial assimilation of algebraic geometric know-how. In particular, we would like to clarify the role  higher secant varieties might play in a wide array of problems, and whether direct physical implications may be inferred from it. 
\section*{Acknowledgements}
The authors wish to thank J.{} Martin and L.{} L.{} S\'anchez-Soto for kindly bringing to their attention several relevant references. They also acknowledge partial financial support from the UNAM-DGAPA-PAPIIT project IG 100316.
\bibliographystyle{plain}
\bibliography{C:/Users/chryss/Documents/chryss_home/papers/strings}   
%
\end{document}